\documentclass[letterpaper,twocolumn,10pt]{article}
\usepackage{usenix,epsfig,endnotes}
\usepackage{latexsym}
\usepackage{graphicx}
\usepackage{mathtools}
\usepackage{amsthm,amssymb,fullpage}
\usepackage{url}
\usepackage{framed}

\newtheorem{lemma}{Lemma}
\newtheorem{theorem}[lemma]{Theorem}

\newtheorem{corollary}[lemma]{Corollary}
\newtheorem{definition}[lemma]{Definition}

\newcommand{\beq}{\begin{equation}}
\newcommand{\eeq}{\end{equation}}
\newcommand{\beas}{\begin{eqnarray*}}
\newcommand{\eeas}{\end{eqnarray*}}

\newcommand{\NL}{\ensuremath{\mathsf{NL}}}

\newcommand{\NCzero}{\ensuremath{\mathsf{NC^0}}}
\newcommand{\NCone}{\ensuremath{\mathsf{NC^1}}}

\newcommand{\bbP}{\mathfrak P}
\newcommand{\bbm}{\mathfrak m}

\newcommand{\junk}[1]{}

\begin{document}

\title{Secure and scalable match: overcoming the universal circuit bottleneck using group programs}


\author{
Rajesh Krishnan\\
        Cosocket LLC\\
        \texttt{krash@cosocket.com}
\and
Ravi Sundaram\\
Northeastern University\\
\texttt{koods@ccs.neu.edu}
}

\maketitle

\begin{abstract}
\label{sec:abstract}
Confidential   Content-Based    Publish/Subscribe   (C-CBPS)   is   an
interaction (pub/sub) model that allows parties to exchange data while
still protecting  their security and privacy interests.  In this paper
we advance the state of the art in C-CBPS by showing how all predicate
circuits  in  \NCone~   (logarithmic-depth,  bounded  fan-in)  can  be
securely   computed   by   a   broker   while   guaranteeing   perfect
information-theoretic  security.  Previous   work  could  handle  only
strictly  shallower  circuits  (e.g.  those  with  depth  $O(\sqrt{\lg
n})$\footnote{We use $\lg$ to denote $\log_2$.})\cite{SYY99, V76}.  We
present three  protocols --  UGP-Match, FSGP-Match and  OFSGP-Match --
all three are based  on (2-decomposable randomized encodings of) group
programs  and handle  circuits  in \NCone.  UGP-Match is  conceptually
simple and has a clean proof  of correctness but it is inefficient and
impractical.  FSGP-Match  uses a ``fixed structure''  trick to achieve
efficiency   and   scalability.   And,   finally,   OFSGP-Match   uses
hand-optimized  group  programs  to  wring  greater  efficiencies.  We
complete  our  investigation  with  an experimental  evaluation  of  a
prototype implementation.
\end{abstract}

\section{Introduction}
\label{sec:intro}
\subsection{Motivation}
\label{subsec:motivation}
Pub/sub systems are an efficient means of routing relevant information
from publishers or content generators to subscribers or consumers. The
efficiency  of pub/sub  models comes  from the  fact  that subscribers
typically receive only a subset  of the total messages published.  The
process of  selecting messages for reception and  processing is called
filtering.    There    are   two    common    forms   of    filtering:
topic-based\footnote{In a  topic-based system, messages  are published
  to  ``topics''   or  named   logical  channels.  Subscribers   in  a
  topic-based system will receive all messages published to the topics
  to which they subscribe, and all subscribers to a topic will receive
  the  same messages. The  publisher is  responsible for  defining the
  classes  of  messages  to  which  subscribers  can  subscribe.}  and
content-based.    In   this    paper   we   focus   on   Content-Based
Publish/Subscribe systems (CBPS) where  messages are only delivered to
a subscriber  if the  attributes or metadata  of those  messages match
predicates defined  by the  subscriber. The subscriber  is responsible
for specifying his  preferences as a predicate over  the attributes of
the content produced by the  publisher. For the purposes of this paper
we will  assume that the  predicate is expressed  as a circuit  and we
will use the terms predicate  and circuit interchangeably.  There is a
third  party,  the  broker,   who  is  responsible  for  matching  the
subscriber's predicate to the  metadata produced by the publisher and,
in case of a match,  forwarding the associated data to the subscriber,
see Figure \ref{fig:pubsub} for  the basic interaction pattern (ignore
the encryptions  for now). The loose coupling  between subscribers and
publishers enabled by the broker allows for greater scalability.  CBPS
is an incredibly useful means  of disseminating information and can be
viewed  as an  abstraction  for a  variety  of different  applications
ranging from  forwarding of e-mail to query/response  systems built on
top of databases.

Scalability of CBPS systems and the distributed coordination of a mesh
of  brokers are  natural  concerns.  But, in  recent  times, with  the
proliferation of online social networks  and new forms of social media
an   even  more  pressing   concern  has   come  into   focus,  namely
confidentiality.  {\em Publisher confidentiality} refers to the notion
that  the publisher would  like to  keep his  content secure  from the
broker,  e.g., the  stock  exchange would  like  to keep  ticker/price
information   private   to    prevent   reselling.   {\em   Subscriber
  confidentiality} is  the notion that subscribers would  like to keep
their preferences private, e.g. a  hedge fund would not wish to reveal
their interest  in a particular stock.  The  widespread development of
web  and mobile  apps  has created  a  proliferation of  third-parties
involved  in the  business of  handling user  preferences  and routing
content, e.g. iPhone and Android apps, Facebook and Twitter apps.  But
even in earlier  times, there has always been  the need for preserving
the privacy of both the  publisher and the subscriber.  For example, a
database server  must not  learn what information  was requested  by a
client, and yet  have the assurance that the  client was authorized to
have  the information  that was  sent; a  mail relay  must be  able to
forward the relevant emails without learning the contents of the email
or  the  subscribers of  a  mail-list.  This  motivates the  need  for
Confidential  Content-Based Publish/Subscribe  (C-CBPS)  schemes.  See
Figure \ref{fig:pubsub}  (note that  the publisher's metadata  and the
subscriber's predicate are encrypted).

\begin{figure}
\centering
\includegraphics[width=3.4in, height=3in]{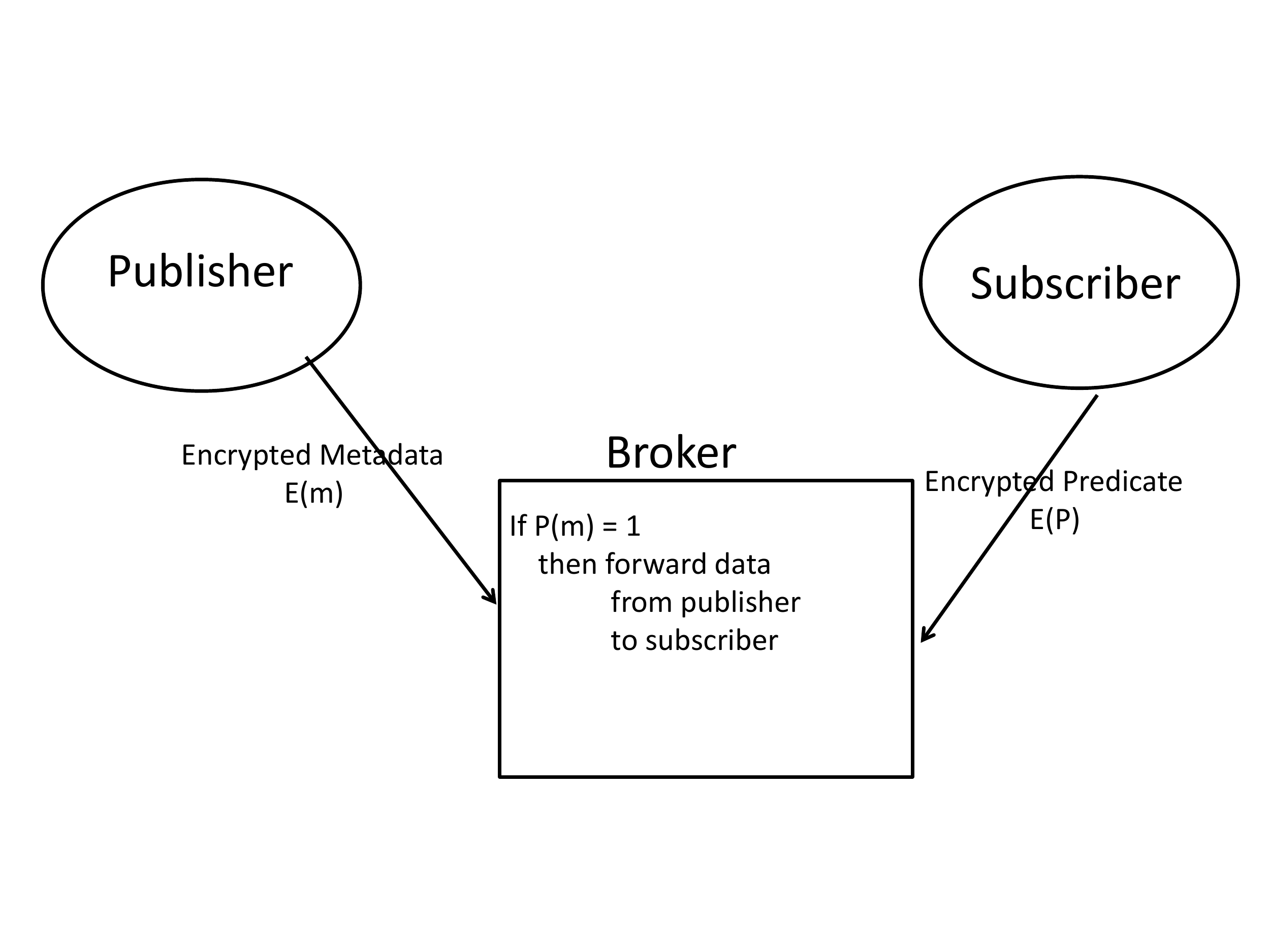}
\vspace{-0.5in}
\caption{The C-CBPS interaction model}
\label{fig:pubsub}
\end{figure}

\subsection{Our Results}
\label{subsec:results}
Until recently, the problem of  securing the privacy of the publishers
and subscribers had  seemed to be a forbidding  task. Practical C-CBPS
systems were able only to  handle the confidentiality of exact matches
and some minor variants \cite{RR06}. Sophisticated schemes that handle
more  expressive  subscriber  predicates  were  too  slow  to  use  in
practice. In general, the  schemes were either practical or expressive
but  not   both.   However,  the  recent   and  dramatic  breakthrough
\cite{G09} in {\em fully homomorphic encryption (FHE)}\footnote{FHE is
  an encryption scheme  that allows a third party  to take encryptions
  $E(a)$ and  $E(b)$ of two  bits, $a$ and  $b$, and obtain  from them
  encryptions  of $E(\bar{a})$,  $E(\bar{b})$, $E(a)  \vee  E(b)$, and
  $E(a) \wedge E(b)$, {\em without  access to the private key used for
    encryption}.}  \cite{G09} has spurred a flurry of activity in this
space.  FHE has  improved  the efficiency  of  {\em Secure  Multiparty
  Computation (SMC)} \cite{KL07} schemes  of which C-CBPS is a special
case.   This   is  an  asymptotic  and   simultaneous  improvement  in
communication  and computational overheads  but the  transformation to
obtain a  C-CBPS scheme  is not  trivial and does  incur some  loss in
efficiency.  Though  substantial progress  is being made  in improving
the speed  of FHE schemes,  it is still  the case today that  fast and
practical realizations are  a ways off. Motivated by  FHE we went back
and took  a second  look at a  scheme that  is nearly two  decades old
\cite{FKN94}.  Reusing the technology  of group programs \cite{B89} in
the  context of  decomposable randomized  encodings \cite{A11}  we are
able not only to obtain a theoretical advance on the state of the art,
but  we  have   in  fact,  produced  a  protocol   that  is  fast  and
practical.  Furthermore, our  main results  guarantee security  in the
{\em  unconditional}  information-theoretic  setting whereas  the  FHE
advance is in the computational setting.

We have two main contributions in this paper.

\begin{itemize}
\item  We   present  an  information-theoretically   secure  protocol,
  UGP-Match  employing {\em  universal  group programs}  to match  any
  predicate  in   \NCone.   UGP-Match  demonstrates   the  theoretical
  possibility of attaining \NCone~ but is not practical.

\item We  then show  how a  ``fixed structure'' trick  gets us substantially 
  closer to
  creating a  practical and real-world protocol,  FSGP-Match, to achieve
  fast and secure  matching of any predicate in  \NCone. FSGP-Match is
  identical to UGP-Match in security guarantees but is much faster. We
  then squeeze  out additional  efficiency using hand  optimized group
  programs  to achieve  our  fastest protocol,  OFSGP-Match, which  is
  closer to being practical.

\end{itemize}

We  have built
prototypes of all  three protocols and created an  implementation of a
real-world pub-sub system. We present the results of our evaluation on
a testbed. The good news is that our schemes provide us with a level
of expressivity (\NCone) that was previously unattainable and in the 
strongest model of security (information-theoretic). The bad news is that
there is still a gap to be overcome to make these protocols truly
practical - we show in Section \ref{sec:implementation} that if  a subscriber and publisher 
wish to compute a secure match based on, say, the Hamming distance (chosen 
as a representative function of interest) between 
their respective (private) bit-vectors then on today's laptops one cannot 
go above length 16 for the bit-vectors and still compute the match under 1s. 
Further, such single message protocols that achieve perfect information
theoretic security in the context of publish/subscribe
incur a tremendous overhead: for every published notification, the publisher
and all subscribers need to prepare fresh messages to send to the broker.

The above  C-CBPS protocols can be  converted to use  shared seeds and
pseudorandom  generators (PSRGs)  rather than  shared  randomness. The
resulting protocols are easily seen  to be secure in the computational
setting based  on the  unpredictability of the  PSRGs.

The focus of this paper is  on achieving a secure match algorithm that
is scalable.   Therefore we concentrated on  bounded depth predicates,
i.e. predicates in \NCone. But  we also have additional asymptotic and
complexity-theoretic  improvements  which   are  not  practical.   For
completeness we mention them here but  we do not present the proofs or
constructions  in this  paper.   One, using  {\em universal  branching
  programs} we  can give an  information-theoretically secure (C-CBPS)
protocol   for  matching  any   predicate  in   \NL  (nondeterministic
logspace).   Second,   using  randomized  encodings  we   can  give  a
\emph{computationally secure} (C-CBPS) protocol to match any predicate
in P (polynomial time).

\subsection{Related Work}
\label{subsec:related}

Several  practical  CBPS systems  have  been  built  for supporting  a
variety of distributed applications.  Siena \cite{CRW01} is one of the
most well known; Gryphon  \cite{BCMNSS99} and Scribe \cite{DGRS03} are
others. Work  on the  security aspects, namely  C-CBPS in  the systems
community is less  than a decade old. A  C-CBPS system supporting only
equality matches  is presented in \cite{SL05} and  a system supporting
extensions   to  inequality   and  range   matches  is   presented  in
\cite{RR06}.  Both  these systems are in the  setting of computational
security.   However,   neither    of   these   systems   satisfy   our
confidentiality  model since  they allow  the broker  to see  that the
encrypted predicates from two different subscribers are identical.

As mentioned earlier,  C-CBPS is a special case of SMC  and there is a
large  and  comprehensive  literature  on  SMC  in  the  cryptographic
community  \cite{G01, G04,  KL07}. The  special case  of  2-party SMC,
known  as Secure Function  Evaluation (SFE),  is an  important subcase
\cite{K09}, that is (different from but) closely related to the C-CBPS
model.   Research  in  SMC  started  nearly 3  decades  ago  with  the
path-breaking work of Yao, \cite{Y82}, and is broadly divided into two
classes  of   protocols  -  those  that   are  computationally  secure
(i.e.  conditioned on  certain  complexity-theoretic assumptions)  and
those that are information-theoretically (or unconditionally) secure.

In the  computational setting the  initial work of Yao  \cite{Y82} and
Goldreich, Micali and  Wigderson \cite{GMW87} has led to  a large body
of  work  \cite{G01,  G04, KL07}.  In  the  past  decade a  number  of
practical  schemes based  on garbled  circuits and  oblivious transfer
have  emerged for the  honest-but-curious adversarial  model: Fairplay
\cite{BNP08}, Tasty  \cite{HKSSW10} and VMCrypt  \cite{M11}. But these
schemes  are  for  SMC   and  are  inefficient  and  impractical  when
restricted to the special case of C-CBPS because of the need to handle
universal circuits \cite{KS08}. Again, as mentioned before, the recent
breakthrough  in FHE  \cite{G09}  holds out  hope  for more  efficient
protocols for SMC but practical protocols are yet to be realized.

On the information-theoretic front  the works of BenOr, Goldwasser and
Wigderson   \cite{BGW88}  and   Chaum,  Cr{\'e}peau   and  D{\aa}mgard
\cite{CCD88}  proved completeness  results for  SMC.  There  have been
additional improvements  \cite{DN07} using the  subsequently developed
notion of randomized encodings  \cite{AIK04, AIK05}.  The seminal work
of Feige, Kilian  and Naor \cite{FKN94}, which lays  the groundwork of
this paper,  considered an interaction model, the  FKN model (detailed
in Definition \ref{def:FKNModel}, SubSection \ref{subsec:terminology})
that is closely  related to, but different from  the C-CBPS model. For
the FKN  model \cite{FKN94}  demonstrated protocols for  predicates in
\NCone~ and  also in \NL.   Though $\NCone \subseteq  \NL$ \cite{AB09,
  P94} the  \NCone~ protocol is the one  we build on in  this paper to
achieve scalable and secure matching in the C-CBPS model.

The protocol of  Feige, Kilian and Naor \cite{FKN94}  in the FKN model
is rendered  impractical when used for  the case of  C-CBPS because of
the need (for the broker) to compute universal circuits \cite{V76}. We
explain this  briefly: the FKN  model involves the broker  computing a
known public  function $f(\bbP,\bbm)$  given the encrypted  version of
$\bbm$, data from the publisher, and $\bbP$, data from the subscriber.
C-CBPS is  the special case  where $f(\bbP,\bbm) =  \bbP(\bbm)$, i.e.,
$\bbP$  is a predicate  (or the  encoding of  one) and  $f$ is  a {\em
  universal } function that simulates  $\bbP$ on $\bbm$.  The need for
universal circuits  poses a barrier - both  theoretical and practical.
Though, some optimizations for universal circuits have been discovered
\cite{KS08}, the  current state of the  art is that  the subscriber is
restricted to predicates of strictly sublogarithmic depth \cite{SYY99,
  V76}  -  more specifically,  the  predicate  is  constrained by  the
condition $d*(\lg  s) = O(\lg n)$ where  $d$ is the depth  and $s$ the
size of  the (bounded fanin)  subscriber predicate. In  general, since
$s$ can be as  large as $2^d$ this means that $d$  is restricted to be
$O(\sqrt{\lg n})$.  In this paper we show how to bypass the need for a
universal circuit and  achieve $d = \Omega(\lg n)$  i.e.  we can match
any \NCone~  predicate in the  C-CBPS interaction model.  We  show how
carefully constructed  2-decomposable randomized encodings \cite{IK00,
  IKOS08,  IKOS09} can be  used to  securely and  efficiently simulate
arbitrary circuits of logarithmic depth.

\section{Preliminaries}
\label{sec:preliminaries}
\subsection{Adversarial Model}
In our C-CBPS  model there are 3 parties --  the broker, the publisher
and the subscriber.  The publisher and subscriber have a shared random
string $r$ known only to the  two of them. Meta-data $\bbm$ is private
to the  publisher and predicate  $\bbP$ is private to  the subscriber.
The publisher and subscriber each  have a separate and private channel
to the  broker. They are to each  send a single message  to the broker
from  which the  broker  should  be able  to  correctly, securely  and
efficiently  compute the  value of  $\bbP(\bbm)$ but  learn absolutely
nothing  else. We  work  in the  information-theoretic security  model
where we don't make  any constraining assumptions on the computational
resources  available  to  the  broker.   However,  the  publisher  and
subscriber are  restricted to be  polynomial-time probabilistic Turing
machines. In fact, given  the shared randomness they are deterministic
polynomial-time  Turing  machines. And  in  keeping with  Kerckhoffs's
principle \cite{Kerckhoff1883} it is assumed that the broker knows the
details  of   the  algorithm/protocol  being  used,   i.e.,  it  knows
everything except for $\bbP$,  $\bbm$ and $r$.  This adversarial model
is captured in Figure \ref{fig:adversarialmodel}

\begin{figure}
\centering
\includegraphics[width=3.4in, height=3in]{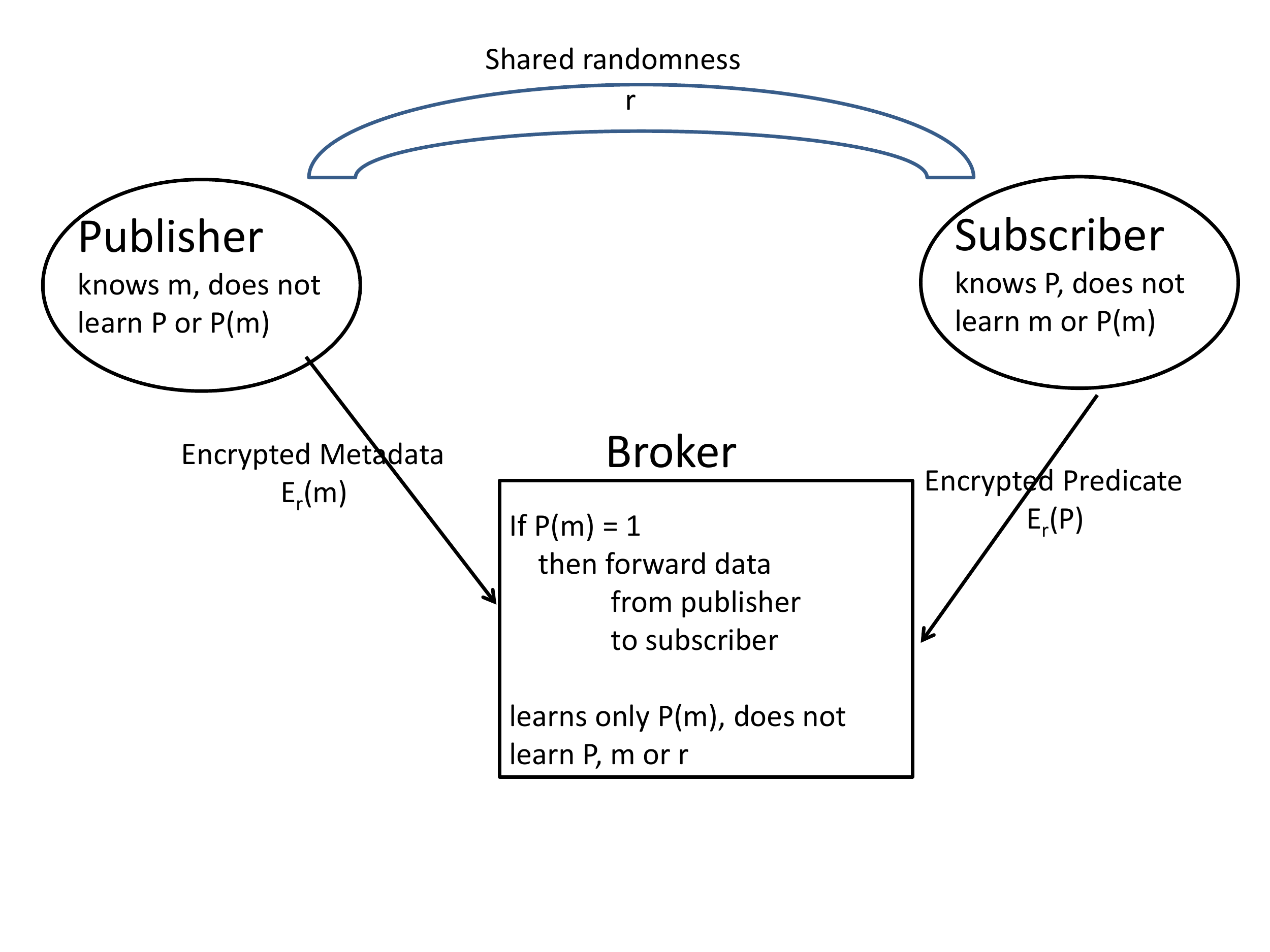}
\vspace{-.5in}
\caption{Adversarial model}
\label{fig:adversarialmodel}
\end{figure}

\subsection{Measure}
We denote the  length of the meta-data $\bbm$ as $n  = |\bbm|$. In the
case of the predicate $\bbP$ the  relevant measure is the depth and we
parameterize it as a multiple of  $\lg n$, to be precise we denote the
depth  of $\bbP$  by $\kappa\lg  n$ where  $n$ is  the  parameter that
asymptotically  goes to  infinity while  $\kappa$ is  a  constant. The
reason   for  parameterizing  in   this  way   is  that   the  obvious
parameterization of the  {\em size} of $\bbP$, in  terms of the number
of gates,  is not relevant  as it is  still an open problem  to handle
arbitrary  polynomial-sized  predicate  circuits  in  the  information
theoretic setting. We remind the reader that the defining contribution
of  this  paper   is  showing  how  to  handle   circuits  in  \NCone,
i.e. logarithmic-depth circuits.

As we will see from the  subsequent sections of this paper, the broker
will get two (sub)sequences of  group elements each from the publisher
and subscriber that he will  interlace and multiply together to obtain
$\bbP(\bbm)$. We  denote by  $L$ the length  of the sequence  that the
broker composes from the  shares he receives. The efficiency question,
thus, becomes  given an $n$  and a $\kappa$  what is the  smallest $L$
that a given  protocol achieves. In what follows we  will see that, in
the  case  of  the  protocol  based  on  Valiant's  universal  circuit
\cite{V76}, $L = n^{\Omega(\kappa\lg n)}$ which is non-polynomial.  In
general  the goal of  this paper  is to  achieve $L$  which will  be a
polynomial in $n$ but the lower  the degree of the polynomial the more
efficient  the protocol.  We will  show that  UGP-Match achieves  $L =
n^{12\kappa+1}\lg  n$.   With  FSGP-Match   we  bring  this   down  to
$L=4n^{2\kappa+2})$ and then finally with OFGSP-Match we bring it down
to $L=2n^{2\kappa+1}$. We point out that these upper bounds are exact,
i.e. we can analyze these constructions down to the exact constant and
so  do not need  to employ  the big-oh  notation.  For  convenience we
present our results in the form of a table.

\begin{table}[h]
\centering
\begin{tabular}{||c|cc||}
\hline \hline
Protocol & Complexity(L) & In words\\ \hline \hline
Universal Circuit & $n^{\Omega(\kappa\lg n)}$ & (super-poly-time)\\ \hline
UGP-Match & $\tilde{O}(n^{20\kappa+2})$ & (poly-time)\\ \hline
FSGP-Match & $4n^{2\kappa+2}$ & '' \\ \hline
OFSGP-Match & $2n^{2\kappa+1}$ & ''\\
\hline \hline
\end{tabular}
\caption{This table shows the complexity of the various protocols.}
\label{tab:complexity}
\end{table}

\subsection{Terminology and Propositions}
\label{subsec:terminology}
We set up some terminology and  definitions for use in the rest of the
paper. We  also state and/or  prove some basic propositions  that will
set the ground for the results to come later.  

For the sake of completeness we define the C-CBPS model formally.
\begin{definition}[C-CBPS]
The 3  parties in the C-CBPS  model and their states  of knowledge are
captured in Figure \ref{fig:adversarialmodel}. There is a single round
of  communication  where the  publisher  and  subscriber  each send  a
private  message ($E_r(\bbm)$  and $E_r(\bbP)$,  respectively)  to the
broker who is then able to efficiently compute $\bbP(\bbm)$ such that

\noindent{\bf Correctness}  $\forall \bbm,  \bbP$, and $r$,  given the
encrypted   messages  $E_r(\bbm),   E_r(\bbP)$  the   broker  computes
$\bbP(\bbm)$ correctly all the time.

\noindent{\bf  Security}  $\forall \bbm,  \bbP$,  and  $r$, given  the
encrypted  messages $E_r(\bbm), E_r(\bbP)$  the broker  learns nothing
whatsoever about $\bbm, \bbP$ (other than the value of $\bbP(\bbm)$).

\end{definition}

We will use the language  of multiplicative groups. For our purposes a
\emph{group} is a  just a set of elements with  a binary operation and
inverses. We  let G  denote a generic  multiplicative group.   We will
need $G$ to be an unsolvable group (for reasons to be explained later)
and so we can take it  to be $S_5$ the symmetric group of permutations
of  a  $5$-element set.   By  default we  will  use  the standard  and
implicit  one-line notation  derived from  Cauchy's  two-line notation
\cite{WikiPermNotation} to represent a permutation.
\begin{definition} [Cycle]
A permutation  $g \in S_5$ is said  to be a \emph{cycle}  if its graph
consists of  exactly one cycle of length  $5$. For example, $(5  3 4 1
2)$  is a  cycle  because its  graph  is the  cycle  $1 \rightarrow  5
\rightarrow 2 \rightarrow 3 \rightarrow 4 \rightarrow 1$.
\end{definition}

When  we use the  product-of-cycles notation  then we  will explicitly
have a subscript ``cycle'' at the end.   For example $(2 4 5 1 3) = (1
2 4)(3 5)_{\mbox{cycle}}$ represents the permutation $1 \rightarrow 2,
2 \rightarrow  4, 3 \rightarrow 5,  4 \rightarrow 1,  5 \rightarrow 3$
with graph $1 \rightarrow 2 \rightarrow 4 \rightarrow 1, 3 \rightarrow
5 \rightarrow 3$.

For the rest of this paper we  will use $\alpha \in S_5$ to denote the
cycle $(23451)$. And, the identity in the group is $1_{S_5} = (1 2 3 4
5)$.

Our  protocols  will involve  group  programs  which  in turn  involve
sequences  of group elements  and their  products. This  motivates the
following:

\begin{definition}[Value of a sequence]
Given a sequence $S$ of  group elements $g_1$, $g_2$, $\ldots$, $g_L$,
the {\em value} of the sequence $\mbox{Value}(S)$ is defined to be the
product of the sequence elements in order, i.e.
\[ \mbox{Value}(S) = \prod_{i=1}^L g_i = g_1g_2\ldots g_L. \]
\end{definition}

\begin{definition}[Blinding]
Given a sequence $S$ of  group elements $g_1$, $g_2$, $\ldots$, $g_L$,
$\mathcal{BS}(S)$ is used to denote the distribution over sequences of
the form
\[ g1\cdot r_1, r_1^{-1}\cdot g_2\cdot r_2,\ldots,  r_L^{-1}\cdot g_L, \] 
generated by choosing each $r_i, \forall_i 1\leq i \leq L-1$ uniformly
and independently from $G$.  We overload the term $\mathcal{BS}(S)$ to
also  refer   to  a  specific  sequence  selected   according  to  the
distribution $\mathcal{BS}(S)$ and the context should be sufficient to
resolve any  ambiguity.  The $r_i$  are referred to as  {\em blinders}
and  {\em blinding}  $S$ refers  to the  act of  selecting  a sequence
according to the distribution $\mathcal{BS}(S)$.
\end{definition}

\begin{lemma}[The Blinding Lemma]
\label{lemma:blindinglemma}
Given a  sequence of group elements,  $S$, of length  $L$, the blinded
sequence $\mathcal{BS}(S)$ has the following two properties:

\noindent{\bf Preserves value}:  $Value(\mathcal{BS}(S)) = Value(S)$

\noindent{\bf  Uniform distribution}:  $\mathcal{BS}(S)$  is uniformly
distributed  over the  space of  all  sequences of  group elements  of
length $L$  with the  same value,  i.e. for any  sequence $S'  = g'_1,
g'_2\ldots g'_L$ we have that
\[ \mbox{Pr}(\mathcal{BS}(S) = S'| \mbox{Value}(S') = \mbox{Value}(S)) = \frac{1}{|G|^{L-1}} \] 
where  the probability  is measured  over  the random  choices of  the
blinders.
\end{lemma}

\begin{proof}
It is  easy to see that  blinding preserves the value  of the sequence
because the blinders,  the $r_i$, cancel out when  the elements of the
blinded sequence are multiplied together.

Now, we need to show the uniform distribution property. First, observe
that  the space  of all  sequences $S'$  such that  $\mbox{Value}(S) =
\mbox{Value}(S')$ has  size exactly  $|G|^{L-1}$.  This is  because we
can pick  the first $L-1$  elements, $g'_1, g'_2,\ldots,  g'_{L-1}$ of
$S'$ arbitrarily from the group  (in $|G|^{L-1}$ ways) but having done
that then  (because these elements are  chosen from a  group) there is
exactly   one   value   for    the   $n$-th   element   $g_L$   namely
$g_{L-1}^{-1}\cdot            g_{L-2}^{-1}            \cdot\ldots\cdot
g_1^{-1}\cdot\mbox{Value}(S)$    such    that   $\mbox{Value}(S')    =
\mbox{Value}(S)$.

Hence, to  compute $\mbox{Pr}(\mathcal{BS}(S) =  S')$ we just  need to
compute  the probability  that the  first  $L-1$ elements  of the  two
sequences  ($\mathcal{BS}(S)$ and $S'$)  match because  the condition,
that  $\mbox{Value}(S')  =  \mbox{Value}(S)$  along with  the  already
proven  fact that  $\mbox{Value}(\mathcal{BS}(S))  = \mbox{Value}(S')$
implies that the $L$-th elements must automatically match if the first
$L-1$ match.

Let $\mathcal{BS}(S) = b_1, b_2,\ldots, b_L$. Then what we have argued
is that
\begin{align*}
  & \mbox{Pr}((b_1 = s'_1) \wedge (b_2 = s'_2) \wedge \dots \wedge (b_L = s'_L)|(\mbox{Value}(S') = \mbox{Value}(S)))\\
  & =\mbox{Pr}((b_1 = s'_1) \wedge (b_2 = s'_2) \wedge \dots \wedge (b_{L-1} = s'_{L-1})| (b_L = s'_L))\\
  & =\mbox{Pr}((b_1 = s'_1) \wedge (b_2 = s'_2) \wedge \dots \wedge (b_{L-1} = s'_{L-1})) \\
  & =\mbox{Pr}(b_1 = s'_1)  \times \mbox{Pr}((b_2 = s'_2)|(b_1 = s'_1))\\ 
  & \times\mbox{Pr}((b_3 = s'_3)|(b_1 = s'_1) \wedge (b_2 = s'_2)) \\ 
  & \times \dots \\
  & \times \mbox{Pr}((b_{n-1} = s'_{n-1})|  (b_1 = s'_1) \wedge (b_2 =
  s'_2) \wedge \ldots \\
  & \hspace{3cm} \ldots \wedge (b_{L-1} = s'_{L-1}))
\end{align*}
But, by the definition of blinding, the above is 
\begin{align*}
  & \mbox{Pr}(r_1 = s_1^{-1}\cdot s'_1) \times \mbox{Pr}(r_2 = s_2^{-1}\cdot s_1^{-1}\cdot s'_1 \cdot s'_2) \\
  & \times \dots \times \mbox{Pr}(r_{L-1} = \prod_{i=1}^{L-1}s_i^{-1}\prod_{i=}^{L-1}s'_i) \\
  & ==\frac{1}{|G|} \times \frac{1}{|G|} \times\dots\times \frac{1}{|G|}\\
  & =\frac{1}{|G|^{L-1}}.
\end{align*}

But this  means that  $\mathcal{BS}(S)$ is uniformly  distributed over
the space of all sequences of group elements with the same value since
the space of all such sequences is of size exactly $|G|^{L-1}$, as has
been argued earlier.

\end{proof}

The following definition, of a group program, is central to this paper
and is presented in visual form in Figure \ref{fig:groupprogram}.

\begin{definition}[Group Program]
Let  $\alpha$  be an  element  of $S_5$.  A  group  program of  length
\emph{$L$}   is  $(g_{1}^{0}$,  $\ldots$,   $g_L^{0})$,  $(g_{1}^{1}$,
$\ldots$ ,$g_L^{1})$, $(k_{1}$, $\ldots$,  $k_L)$ where for any $i,j$:
$g_{i}^{j}  \in S_{5}$ and  $k_{i} \in  \{1,\ldots,n\}$.  We  say that
this program $\alpha$-computes  $f:\{0,1\}^{n} \rightarrow \{0,1\}$ if
$\forall x$,
\[ f(x)=1 \Rightarrow \prod ^{\ell} _{i=1} g_{i}^{x_{k_i}} = \alpha \]
\[ f(x)=0 \Rightarrow \prod ^{\ell} _{i=1} g_{i}^{x_{k_i}} = 1_{S_5}; \]
which  we can write  compactly as  $\forall x  : \prod  ^{\ell} _{i=1}
g_{i}^{x_{k_i}} = \alpha^{f(x)}$.
\end{definition}

\begin{figure}
\centering
\includegraphics[width=3.4in, height=3in]{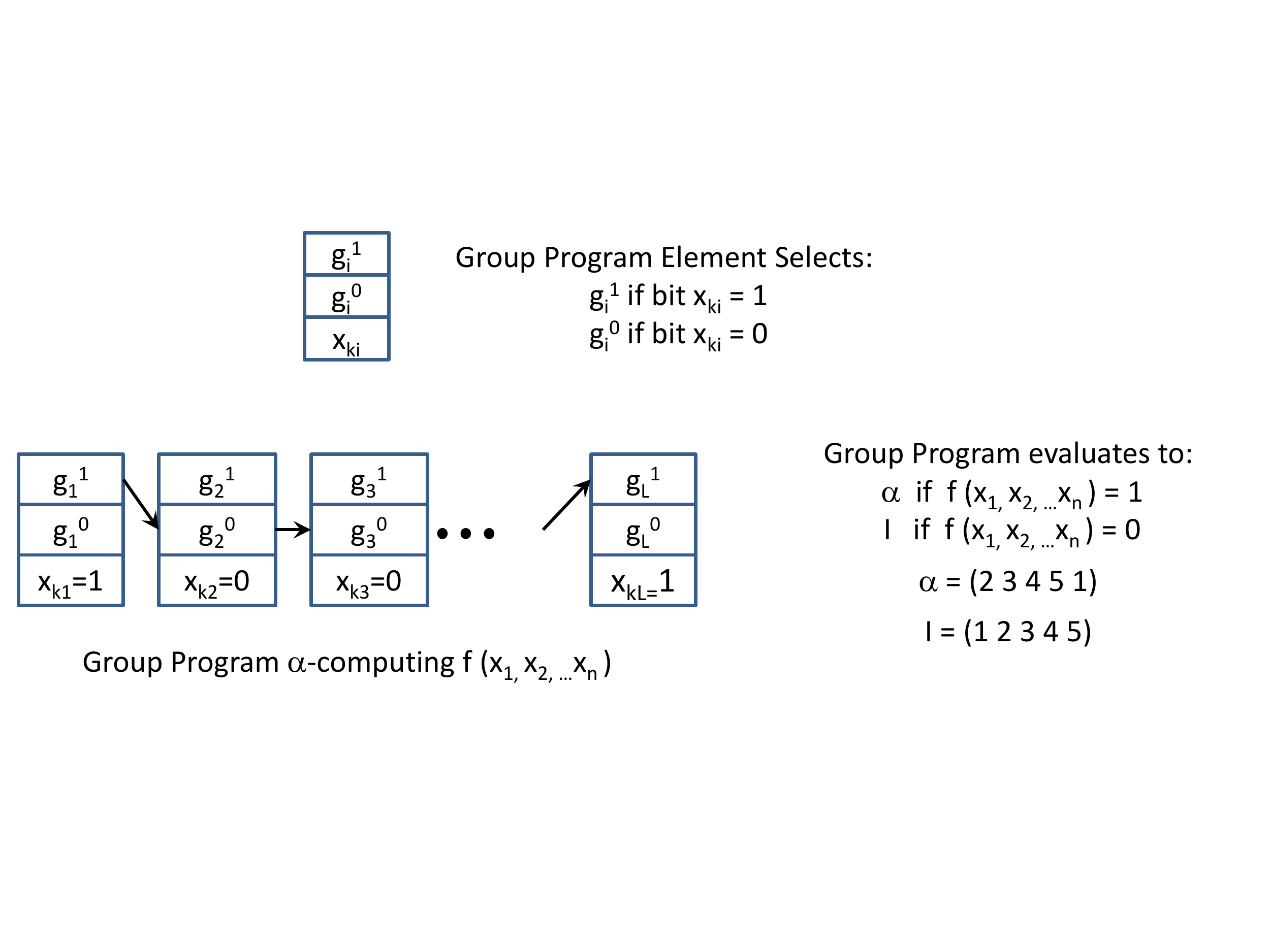}
\vspace{-0.8in}
\caption{The form of a Group Program}
\label{fig:groupprogram}
\end{figure}

We will  consider predicates  represented as $n$-input,  single output
(bounded  fan-in)  circuits   of  AND($\wedge$  two-input),  OR($\vee$
two-input) and  NOT($\neg$ single input) gates. Our  definition of the
depth  of  a  circuit  is  slightly non-standard  in  that  we  ignore
NOT($\neg$) gates. This is because  of the Barrington Transform (to be
elaborated below) which transforms a circuit into a group progam whose
length  depends  only  on  the  depth  of  the  circuit  in  terms  of
AND($\wedge$) and OR($\vee$) gates.

\begin{definition}[Depth of a Circuit]
The depth of a circuit is defined to be the number of AND and OR gates
in the  longest path from  an input to  the output. (NOT gates  do not
count towards depth).
\end{definition}

In his seminal paper \cite{B89}, on the way to showing that \NCone~ is
computable by  fixed-width branching programs,  Barrington showed that
any   logarithmic-depth   circuit    can   be   transformed   into   a
polynomial-length  group  program  -  a  transformation  we  term  the
Barrington Transform.

\begin{theorem}[Barrington Transform] 
\label{thm:barrington-transform}
Any circuit  of depth $d$ can be
  transformed   into  a   group   program  of   length  $4^{d}$   that
  $\alpha$-computes the same function as the circuit.
\end{theorem}

For  the  sake  of  completeness  we present  the  proof  in  Appendix
\ref{appendix:barrington-transform}.   The  proof,  presented  as  a  series  of
lemmas, details  the Barrington Transform  showing how to  transform a
circuit  into a  group program.  This  treatment is  directly based  on
\cite{Viola09}.

It  follows  from   Theorem  \ref{thm:barrington-transform}  that  the
Barrington Transform transforms any $n$-input single output circuit of
depth $\kappa\lg n$ into a group program of length $n^{2\kappa}$ where
$\kappa$ is a  constant.  

We will  be using an extension  of the notion  of randomized encodings
\cite{IK00,  AIK05,  A11}  that  played  a  significant  role  in  the
breakthrough  showing  the  feasibility  of cryptography  in  \NCzero~
\cite{AIK04}.  First, we give the definition of randomized encodings.

\begin{definition}[Randomized encoding] 
A   function  $f(x)$   is   said  to   have   a  randomized   encoding
$\hat{f}(x,r)$,  where $r$  is a  random  string, if  there exist  two
efficiently computable (deterministic, polynomial-time) algorithms REC
and SIM such that

\noindent{\bf  Correctness} $\forall x,  r$, given  $\hat{f}(x,r)$ REC
recovers $f(x)$, i.e. REC$(\hat{f}(x,r)) = f(x)$ .

\noindent{\bf  Security}  $\forall  x$,  given $x$  and  $r'$  (random
coins), SIM produces a distribution identical to $\hat{f}(x,r)$, i.e.,
the distribution of  SIM$(x, r')$ is identical to  the distribution of
$\hat{f}(x,r)$.

\end{definition}

The above definition will make  it easier to comprehend the definition
of a {\em 2-decomposable randomized encoding}. We believe that, though
natural  in the  context of  2-party secure  computation, this  is the
first   time  this   definition   is  appearing   explicitly  in   the
literature. A stronger definition of decomposable randomized encodings
has  appeared   before  \cite{IKOS08,   IKOS09}  and  the   notion  of
2-decomposable  randomized  encodings  is  implicit  in  prior  works,
including \cite{FKN94}.

\begin{definition}[2-decomposable Randomized encoding] 
\label{def:2-decomp}
A  function  $f(x,y)$ is  said  to  have  a 2-decomposable  randomized
encoding $\langle\hat{f}_1(x,r),  \hat{f}_2(y,r)\rangle$, where $r$ is
a  ({\em  shared})  random  string,  if there  exist  two  efficiently
computable  (deterministic, polynomial-time)  algorithms  REC and  SIM
such that

\noindent{\bf    Correctness}    $\forall     x,    y,    r$,    given
$\langle\hat{f}_1(x,r), \hat{f}_2(y,r)\rangle$  REC recovers $f(x,y)$,
i.e., REC$(\langle\hat{f}_1(x,r), \hat{f}_2(y,r)\rangle) = f(x,y)$.

\noindent{\bf Security} $\forall x, y$,  given $x, y$ and $r'$ (random
coins),     SIM    produces     a     distribution    identical     to
$\langle\hat{f}_1(x,r), \hat{f}_2(y,r)\rangle$, i.e., the distribution
of   SIM$(x,   y,  r')$   is   identical   to   the  distribution   of
$\langle\hat{f}_1(x,r), \hat{f}_2(y,r)\rangle$.
\end{definition}

The schemes  presented in this paper  rely crucially on  the model and
construction presented  in \cite{FKN94}. We  define the FKN  model and
show how protocols for it are essentially equivalent to 2-decomposable
randomized encodings.

\begin{definition}[FKN model]
\label{def:FKNModel}
The 3  parties in the FKN  model and their states  of knowledge are
captured  in Figure  \ref{fig:FKNModel}. There  is a  single  round of
communication  where  Alice  and  Bob  each  send  a  private  message
($E_r(x)$ and  $E_r(y)$, respectively)  to Carol who  is then  able to
efficiently compute the publicly known function $f(x,y)$ such that

\noindent{\bf  Correctness}  $\forall  x,   y$,  and  $r$,  given  the
encrypted messages $E_r(x),  E_r(y)$ Carol computes $f(x,y)$ correctly
all the time.

\noindent{\bf Security}  $\forall x, y$, and $r$,  given the encrypted
messages $E_r(x), E_r(y)$ Carol learns nothing whatsoever about $x, y$
(other than the value of $f(x,y)$).
\end{definition}

\begin{figure}
\centering
\includegraphics[width=3.4in, height=3in]{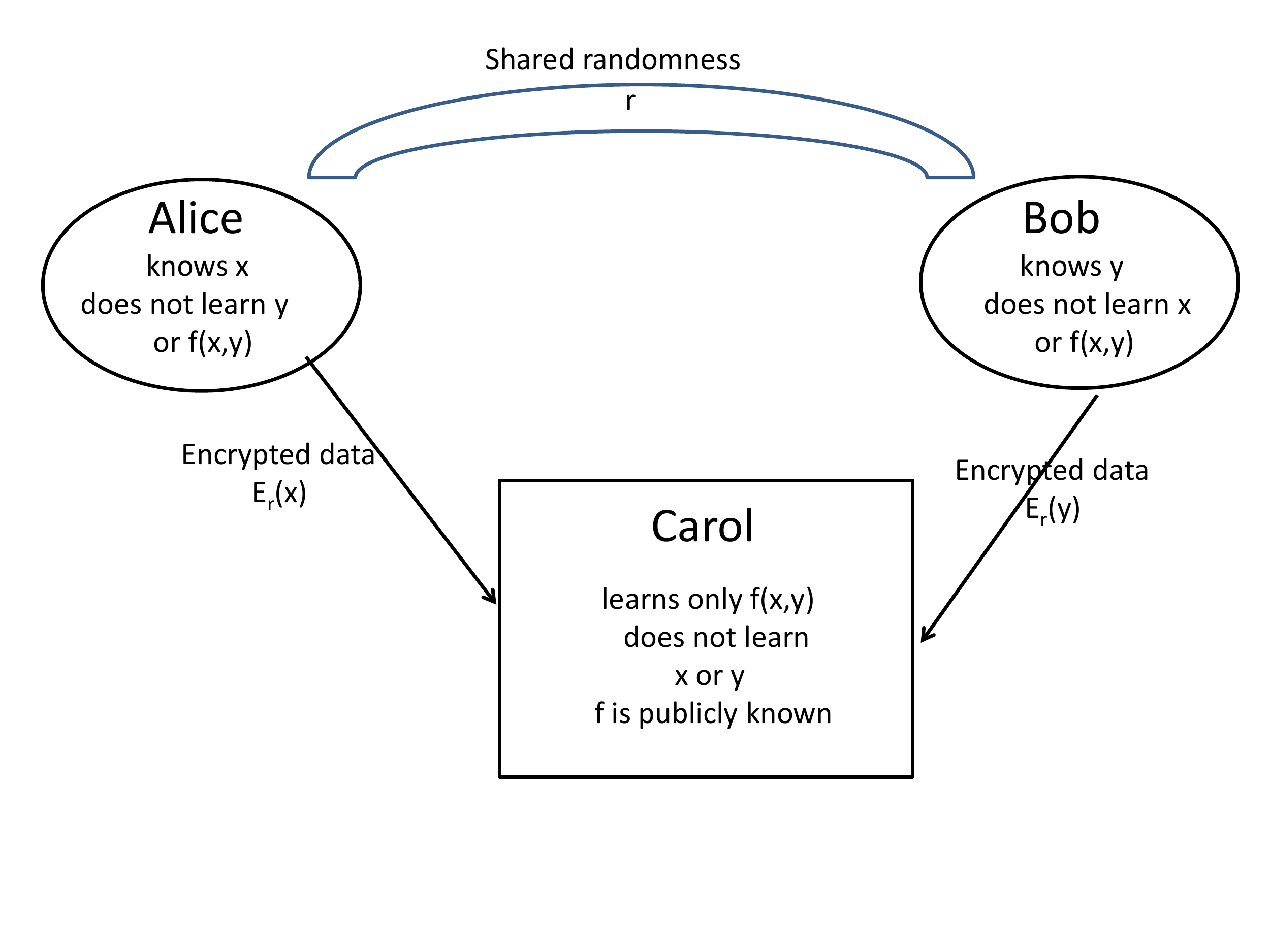}
\vspace{-.5in}
\caption{FKN model}
\label{fig:FKNModel}
\end{figure}

\begin{lemma}[Equivalence   of  FKN   and   2-decomposable  randomized
  encoding]
\label{lemma:FKN2decomp}
There is  a 1-1  isomorphism between protocols  for the FKN  model and
2-decomposable randomized encodings.
\end{lemma}
\begin{proof}
  We will first  see that a 2-decomposable randomized  encoding of the
  function $f(x,y)$  gives rise to a  protocol for the  FKN model. Let
  Alice compute and send  $\hat{f}_1(x,r)$ to Carol while Bob computes
  and send $\hat{f}_2(y,r)$ to Carol. From the Correctness property it
  follows    that    Carol    can   run    REC$(\langle\hat{f}_1(x,r),
  \hat{f}_2(y,r)\rangle)$  to  compute   $f(x,y)$  correctly  all  the
  time.  And from  the  Security  property we  get  that Carol  learns
  nothing but the value $f(x,y)$.

  Similarly,  given a  protocol for  the FKN  model that  computes the
  publicly  known  function  $f(x,y)$  it  is easy  to  see  that  the
  Correctness and Security properties  carry over by setting $E_r(x) =
  \hat{f}_1(x,r)$ and $E_r(y) = \hat{f}_2(x,r)$.
\end{proof}

We present  the definition of  a Universal Function below.  

\begin{definition}[Universal Function]
A Universal Function $\mathcal{UF}$ has  two inputs $x$ and  $y$ where
$y$ is  the encoding of a function  for which $x$ is  a suitable input.
The output  of $\mathcal{UF}(x,y)$ is  defined to be $y(x)$.
\end{definition}

One  can similarly  define  Universal Circuits  which  are similar  to
Universal Functions  except that they  work with encoding  of circuits
and simulate circuits.
\begin{definition}[Universal Circuit]
A Universal Circuit $\mathcal{UC}$ has  two inputs $x$ and  $y$ where
$y$ is  the encoding of a circuit  for which $x$ is  a suitable input.
The output  of $\mathcal{UC}(x,y)$ is  defined to be $y(x)$,  in other
words  the   Universal  Circuit  outputs  the   result  obtained  from
simulating  function $y$  on  input $x$.  (Universal  Circuits are  the
circuit equivalent of Universal Turing Machines).
\end{definition}

The central idea  of this paper is to bypass  Universal Circuits so we
will not  be utilizing this  definition except to explain  (in Section
\ref{sec:overview})  how it  is that  we pick  up efficiency  gains by
avoiding them.

One  can naturally  apply the  notion of  a  2-decomposable randomized
encoding to a universal function $f(x,y) = y(x)$.

And   analogous    to   the   correspondence    expressed   in   Lemma
\ref{lemma:FKN2decomp} we have the following correspondence.

\begin{lemma}[Equivalence  of  C-CBPS  and  2-decomposable  randomized
  encoding of universal functions]
\label{lemma:ccbps-2decompuniv}
There is a 1-1 isomorphism  between protocols for the C-CBPS model and
2-decomposable randomized encodings of universal functions.
\end{lemma}
\begin{proof}
The proof is very similar  to that of Lemma \ref{lemma:FKN2decomp} and
so  we provide  only  the key  aspects  of the  correspondence as  the
details  are  straightforward.   $x,  y$ correspond  to  $\bbm,  \bbP$
respectively.     $\hat{f}_1(x,r),   \hat{f}_2(y,r)$    correspond   to
$E_r(\bbm), E_r(\bbP)$ respectively.  And the Correctness and Security
properties carry over in a direct way.
\end{proof} 

\begin{figure}
\centering
\includegraphics[width=3.4in, height=3in]{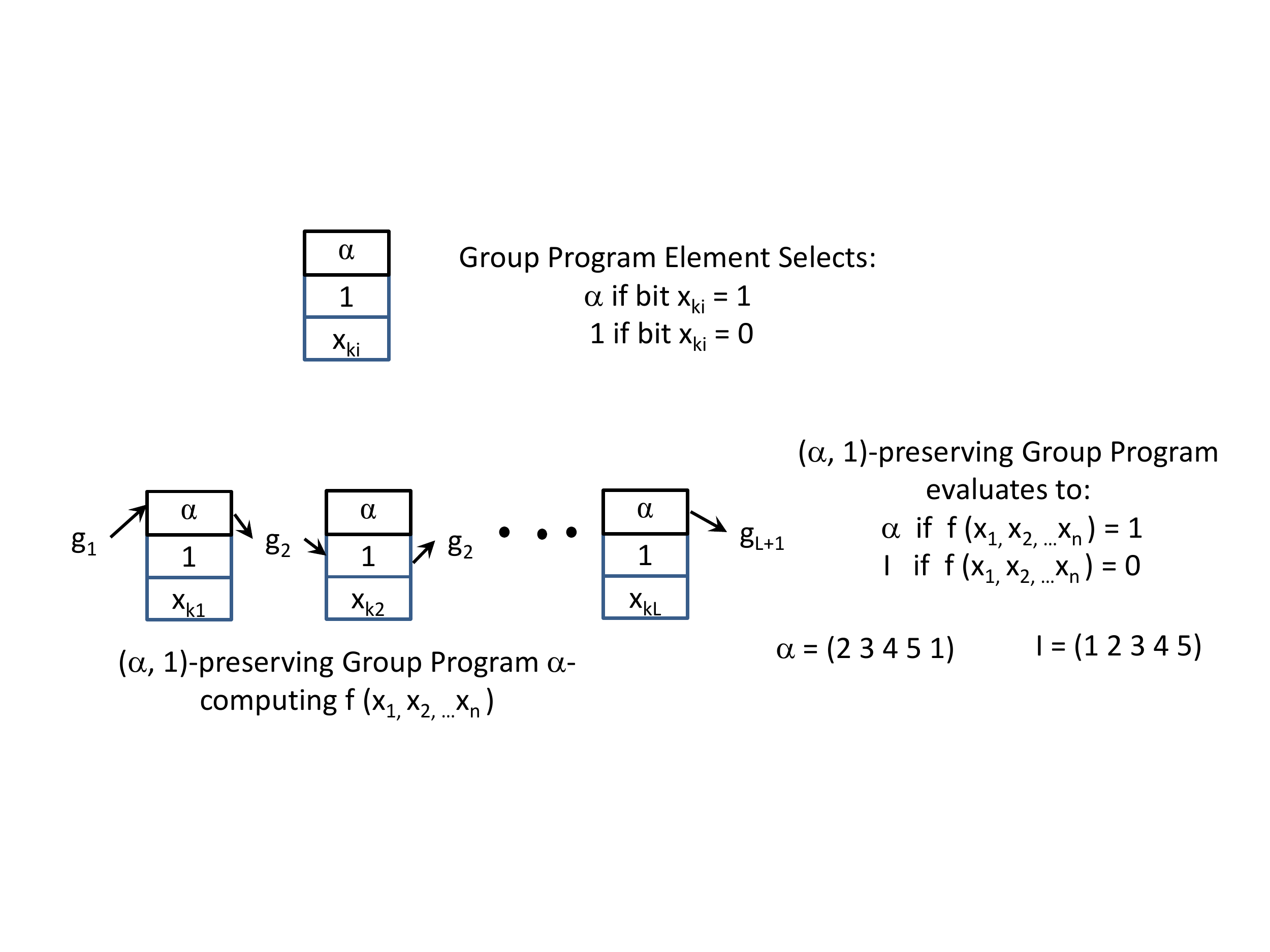}
\vspace{-.8in}
\caption{The form of a $(\alpha, 1)$-preserving Group Program}
\label{fig:alphaone}
\end{figure}

We now  present some definitions that clarify  the ``fixed structure''
trick  alluded  to  earlier,  and  how it  relates  to  2-decomposable
randomized encodings of universal functions.  The following definition
of a $(\alpha,1)$-preserving group program is presented in visual form
in Figure \ref{fig:alphaone}.

\begin{definition}[$(\alpha,1)$-preserving group programs]
  An  $(\alpha,1)$-preserving group  program of  length $L$  is $(g_1,
  g_2,\ldots, g_{L+1})$,  $(k_1, k_2,\ldots, k_L)$ where  for any $i$:
  $g_{i} \in S_{5}$ and $k_{i}  \in \{1,\ldots,n\}$.  We say that this
  $(\alpha,1)$-preserving      group     program     $\alpha$-computes
  $f:\{0,1\}^{n} \rightarrow \{0,1\}$ if $\forall x$,
  \[      f(x)=1      \Rightarrow      (\prod      ^{\ell}      _{i=1}
  g_{i}\cdot\alpha^{x_{k_i}})g_{L+1} = \alpha \]
\[       f(x)=0      \Rightarrow      (\prod       ^{\ell}      _{i=1}
g_{i}\cdot\alpha^{x_{k_i}})g_{L+1}  =  1_{S_5}; \]  which  we can  write
compactly    as     $\forall    x    :     (\prod    ^{\ell}    _{i=1}
g_{i}\cdot\alpha^{x_{k_i}})g_{L+1} = \alpha^{f(x)}$.
\end{definition}

\begin{definition}[Index sequence  of an $(\alpha,1)$-preserving group
  program]
  Given  $(g_1, g_2,\ldots,  g_{L+1})$, $(k_1,  k_2,\ldots,  k_L)$, an
  $(\alpha,1)$-preserving  group  program  of  length  $L$,  its  index
  sequence is $(k_1, k_2,\ldots, k_L)$.
\end{definition}

Analogous   to   Theorem   \ref{thm:barrington-transform} we have

\begin{theorem}[$(\alpha,1)$-preserving Barrington Transform] 
\label{thm:alpha-one-barrington-transform}
Any   circuit    of   depth   $d$   can   be    transformed   into   a
$(\alpha,1)$-preserving   group  program   of   length  $4^{d}$   that
$\alpha$-computes the same function as the circuit.
\end{theorem}

For  the  sake  of  completeness  we present  the  proof  in  Appendix
\ref{appendix:alpha-one-barrington}. 

The  definition of  a  fixed structure group  program  given below  is
crucial to  improving the  efficiency of the  ideas presented  in this
paper and making them practical.

\begin{definition}[Fixed structure group program]
\label{def:fixed-structure}
  If a class  of functions (say, the class  of $n$-input single output
  functions computable by circuits of  depth $\kappa\lg n$) can all be
  transformed  into $(\alpha,1)$-preserving  group  programs with  the
  exact same index sequence (and hence the exact same length) then the
  resulting  class  of  group  programs   is  said  to  have  a  fixed
  structure.  In a  slight abuse  of language  we refer  to  the class
  itself as a fixed structure group program.
\end{definition}

\section{Overview}
\label{sec:overview}
Before explaining the bottleneck that this paper has overcome we first
briefly  sketch   the  protocol  for   the  FKN  model   presented  in
\cite{FKN94}  for  publicly  known  functions $f(x,y)$  computable  as
logarithmic-depth  circuits. The  group program  equivalent of  $f$ is
computed (using  Barrington's Transform \cite{B89}) by  both Alice and
Bob from  the corresponding circuit.  Each of  them instantiates their
share of  the group program based  on their respective  input ($x$ for
Alice  and $y$  for Bob).   Then each  of them  blinds their  share in
coordinated  fashion using  the shared  randomness.  Finally,  each of
Alice and Bob sends their  respective blinded shares to Carol who puts
them  together to  form the  final  blinded sequence  whose value  she
computes   to  obtain   $f(x,y)$.   See  Figure~\ref{fig:FKN}.    What
\cite{FKN94}  essentially   demonstrate  is  the   construction  of  a
2-decomposable randomized encoding $\hat{f}$  of $f$, as elaborated in
Lemma \ref{lemma:FKN2decomp}.  Of course, the notion of 2-decomposable
randomized  encodings  arose much  later  in  the  work of  Ishai  and
Kushilevitz \cite{IK00, IKOS08, IKOS09},  but it gives us a convenient
language  to  think  about  such protocols.   The  individual  shares,
constructed and, sent to Carol by Alice and Bob are just the two parts
of the 2-decomposable randomized encoding, namely $\hat{f}_1(x,r)$ and
$\hat{f}_2(y,r)$.  REC guarantees that Carol is able to learn $f(x,y)$
while SIM guarantees that he learns nothing beyond that.

\begin{figure}
\centering
\includegraphics[width=3.4in, height=3in]{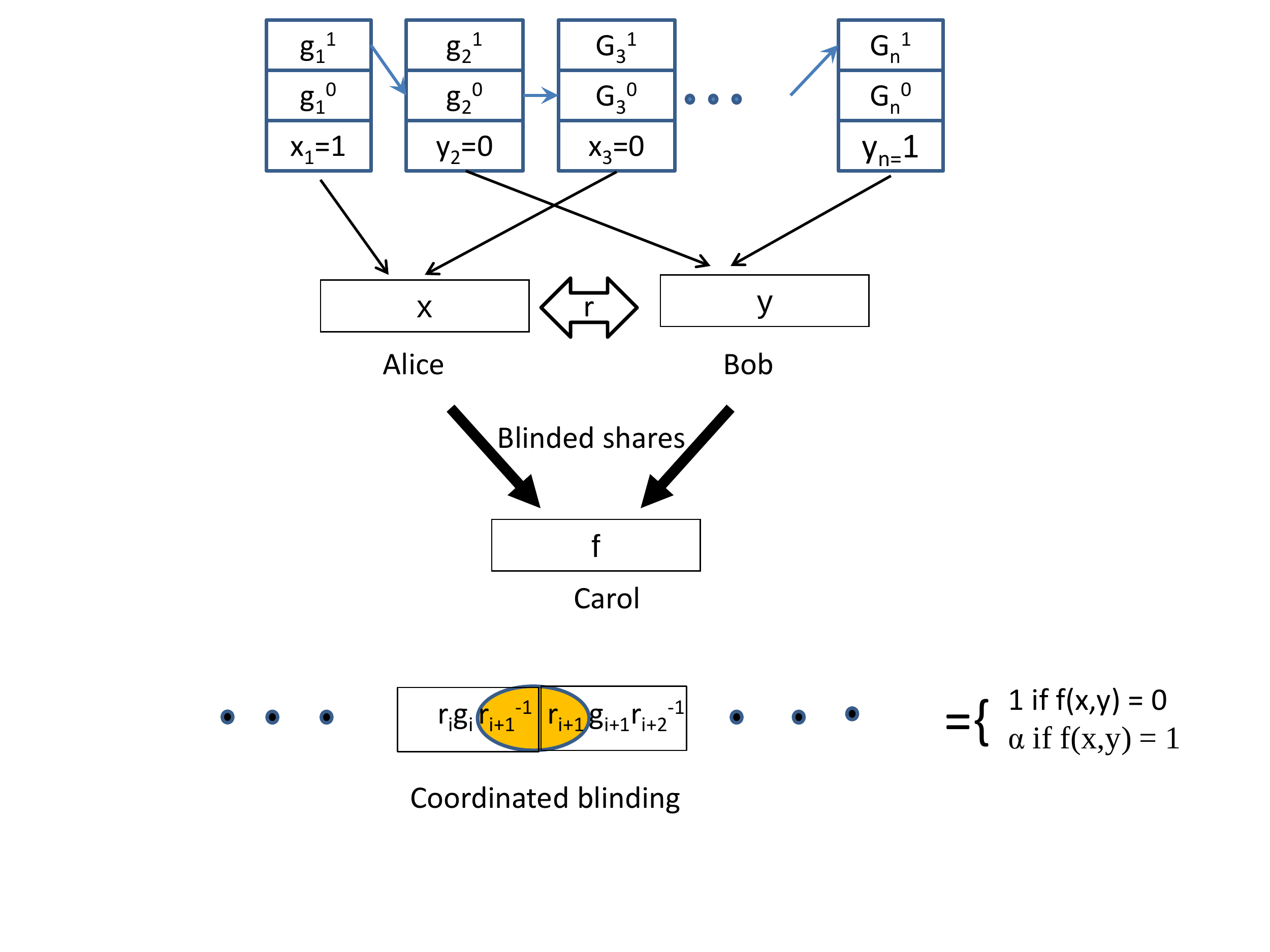}
\vspace{-.5in}
\caption{Protocol for FKN model, see \cite{FKN94}.}
\label{fig:FKN}
\end{figure}

Now, we explain the \emph{universal circuits} bottleneck that we claim
to have overcome. Recall that in  our (C-CBPS) setting $f$ is not just
any function  but it is a  \emph {universal function}  where $f(x,y) =
y(x)$.   The term  \emph{universal} comes  from the  fact that  $f$ is
effectively simulating $y$ on $x$ and so the natural question arises -
how efficiently  can this simulation  be done?  In other  words, given
that we  are restricted to  have only polynomial-size  group programs,
what restriction does  this put on the class  of functions represented
by $y$?  The  best known construction of universal  circuits is due to
Valiant \cite{V76}  who shows that if  $y$ were a circuit  of size $s$
and depth  $d$ then it must  satisfy $d \log  s = O(\log n)$  (for the
resulting  Barrington Transform to  produce a  polynomial-length group
program).  Observe  that this constraint on $d$  and $s$ automatically
prevents $y$ from  representing circuits in \NCone because  $s \geq d$
and hence  $d$ is forced to  be $o(\log n)$.  Sanders,  Young and Yung
\cite{SYY99}  who utilize  Valiant's universal  circuits construction,
mention  that $y$  could  be  the class  of  functions represented  by
circuits of depth $\sqrt{\log n}$  and size $2^{\sqrt{\log n}}$ - note
that this is believed to be a strict subset of the class of predicates
in \NCone.

Our main  contribution is in showing  how we can  bypass the universal
circuits  bottleneck and  instead use  group programs  to  improve the
efficiency  of simulation.   We present  two main  constructions.  The
first, UGP-Match,  which is  primarily of theoretical  interest, shows
how we can handle all of \NCone by encoding the subscriber's predicate
as a group program (using the Barrington Transform) and constructing a
universal group  program using the Barrington  Transform (again). This
construction  has a  conceptually clean  proof but  it  is impractical
since the double  invocation of the Barrington Transform  results in a
final group program  whose length is a very  high-degree polynomial in
$n$.   Our  second  construction,  FSGP-Match,  which  is  potentially
practical, uses a fixed structure trick to avoid one invocation of the
Barrington Transform.  Rather than construct a universal group program
using the Barrington Transform we directly construct a fixed structure
group program, which in essence  is a universal group program but much
more  economical length-wise. By  avoiding the  inefficiencies arising
from use of  the Barrington Transform we obtain  a final group program
whose length  is a relatively  low-degree polynomial in $n$.   In this
second construction  of FSGP-Match we  need a special group  program -
the \emph{selector group program} - which we construct by applying the
Barrington Transform to a \emph{selector circuit}. We also show how an
additional  optimization  can be  achieved  by  hand-crafting the  the
selector  group  program to  obtain  our  most efficient  construction
OFSGP-Match.

\section{Protocols and proofs}
\label{sec:protocols}
\subsection{UGP-Match}
As  explained  in Section  \ref{sec:overview}  the  main  idea in  the
UGP-Match construction  is to use  the Barrington Transform  to encode
the predicate $\bbP$  (which is representable as a  circuit in \NCone)
as a  polynomial-length group program  that indexes into  the metadata
bits $\bbm$.  The (publicly known) circuit $f$ in the protocol for the
FKN model \cite{FKN94} (see Figure  \ref{fig:FKN}) is chosen to be such
that  at the  lowest  level  it first  selects  the appropriate  group
program element based on the  index and value of the corresponding bit
in $\bbm$ and  having selected all the group  program elements it then
multiplies  them  together using  a  standard  divide  and conquer  or
balanced     binary     tree     based    approach,     see     Figure
\ref{fig:select-multiply}. In  a nutshell, UGP-Match uses  the protocol from
\cite{FKN94} for the FKN model, with $f$ being a circuit that takes as
input the metadata $\bbm$ and  the predicate $\bbP$ encoded as a group
program  and outputs  the result  of simulating  $\bbP$ on  $\bbm$. We
present  UGP-Match  more  formally  below.  We refer  to  the  circuit
representing $f$ as the UGP-Match Circuit. We refer to the final group
program  that the  broker  assembles as  the  Universal Group  Program
because it  is essentially a  Group Program that simulates  the group
program representing $\bbP$.

\begin{framed}
\underline{UGP-Match}
\begin{enumerate}
\item  The  publisher and  subscriber  register  with  the broker  the
  precise  form of  their inputs.   In particular  the  publisher must
  specify the number $n$ of bits of $\bbm$ the meta-data (if there are
  fewer relevant  bits, the  remaining bits can  be padded  with dummy
  bits).  And the subscriber specifies $L = n^{2\kappa}$ the length of
  the group  program (in terms of  number of group  element pairs (not
  bits) along with  the bit of $\bbm$ that each  group element pair is
  dependent on).  As before the group program can be padded with dummy
  pairs  if  there are  fewer  relevant  pairs.   We may  assume  that
  everybody is coordinated on the  choice of the non-solvable group in
  which to  carry out their  computations, say $S_5$,  the (symmetric)
  group of permutations on $5$  elements, which itself is a group with
  $120$. For the purpose of this specific cprotocol we can assume each
  element  is specified  in unary  using $5$  bits so  that  any given
  element of the group is specified using $5 \times 5 = 25$ bits.

\item  The broker  now computes  the UGP-Match  Circuit $f$  (which is
  essentially a Select block followed by a divide-and-conquer Multiply
  block,  see Figure  \ref{fig:select-multiply}. It  then  applies the
  Barrington  Transform  to create  the  corresponding group  program.
  Each group element pair in this group program is dependent either on
  a subscriber bit  or on a publisher bit.  The  broker hands back the
  entire group program to both the publisher and the subscriber.

\item The publisher and subscriber  know which pairs belong to each of
  them.     They   have    already   coordinated    their   pre-shared
  randomness. Depending on the value of their individual bit they pick
  the   corresponding  element  of   the  pair   and  then   blind  it
  appropriately. They  then give  their respected blinded  elements to
  the broker.

\item The broker  puts all the blinded elements  together in the right
  sequence and multiplies them.  If  he gets $1$ he forwards along the
  (encrypted)  data from  the  publisher to  the  subscriber, else  he
  withholds it.
\end{enumerate}
\end{framed}

\begin{figure}
\centering
\includegraphics[width=4in, height=3in]{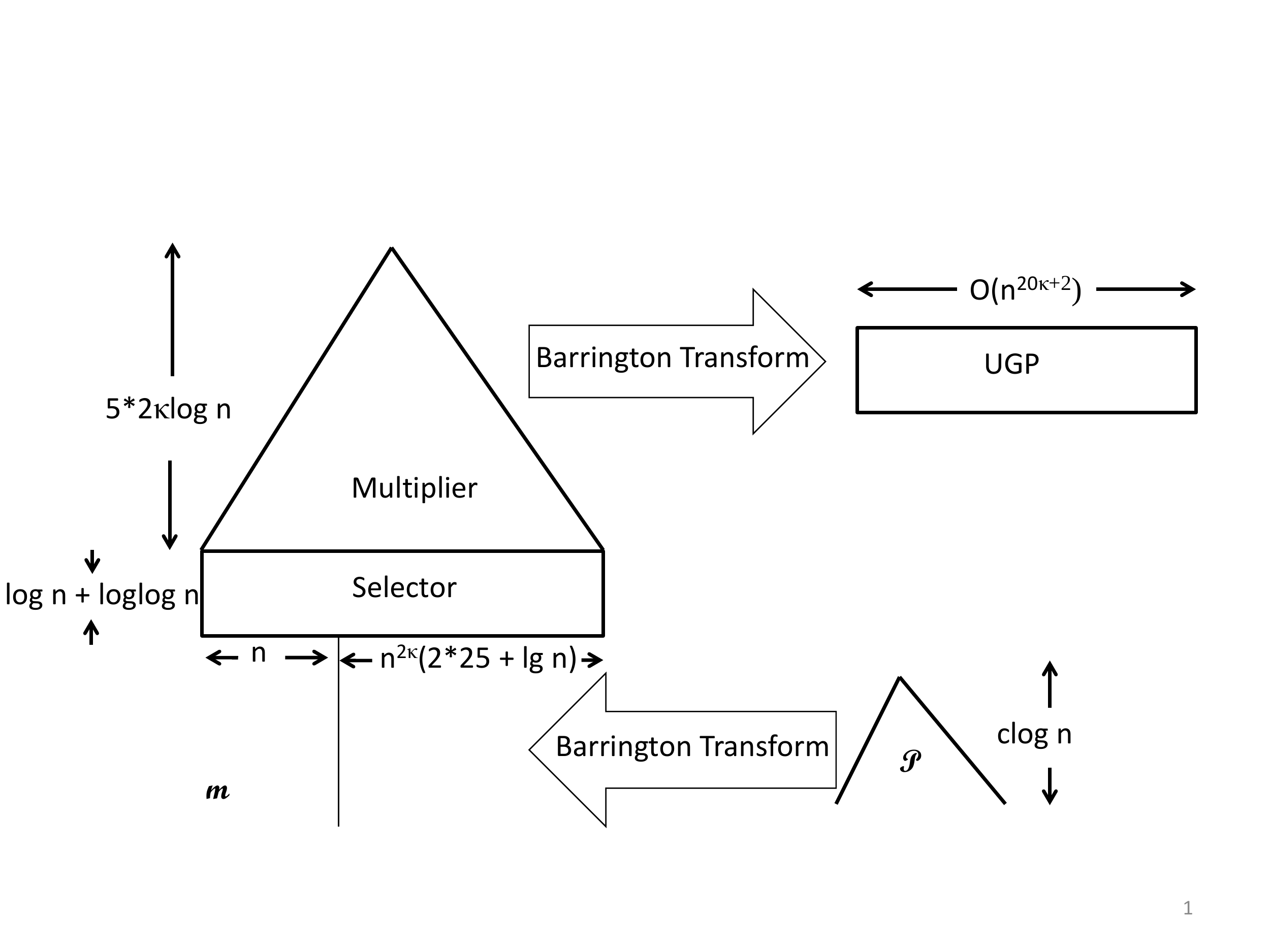}
\vspace{-.4in}
\caption{UGP-Match Circuit}
\label{fig:select-multiply}
\end{figure}

\begin{theorem}
  Given metadata  of size at most  $n$ and predicate of  depth at most
  $\kappa\lg  n$,  UGP-Match  is an  information-theoretically  secure
  protocol     for     the     C-CBPS    model     with     complexity
  $\tilde{O}(n^{20\kappa+2}$.
\end{theorem}
\begin{proof}
  The proof that UGP-Match is correct and secure follows directly from
  the proof in \cite{FKN94} for the FKN model.

  All that remains to do is  to bound the complexity of UGP-Match.  We
  now  provide a detailed  description of  the UGP-Match  Circuit $f$.
  See  Figure  \ref{fig:select-multiply}.  This  circuit  takes in  as
  input  $\bbm$  and  the  bit-representation  of  the  group  program
  obtained  from  applying   the  Barrington  Transform  (see  Theorem
  \ref{thm:barrington-transform}) to  $\bbP$.  The group  program is a
  sequence of $L  = n^{2\kappa}$ group program elements  each of which
  is two group  elements and an index (into  $\bbm$.  We represent the
  group elements  which are permutations  of $S_5$ in unary,  e.g., we
  would represent the  permutation $( 2 3 4 5 1)$  as the bit sequence
  $00010\;00100\;01000\;10000\;00001$  (the spacings  are  placed  for
  convenience  of  reading).   We  note  that this  is  not  the  most
  economical  representation in terms  of bit-length  but what  we are
  ultimately  looking to  minimize is  the length  of the  final group
  program,  i.e. the  depth  of  the UGP-Match  Circuit  and for  that
  purpose  this  representation is  close  to  optimal.   Each of  the
  indices in binary is represented  using $\lg n$ bits.  The UGP-Match
  Circuit is a Select  block followed by a divide-and-conquer Multiply
  block.   In the Select  block, corresponding  to each  group program
  index  the value  of the  bit  of $\bbm$  is extracted  using a  mux
  (multiplexer, see \cite{WikiMux}) which  is a circuit of depth $\lg
  n  + \lg\lg  n$.  The  corresponding group  program element  is then
  selected  using  a circuit  of  depth  $2$.   Multiplying two  group
  program elements takes a circuit of depth $5$ and hence the Multiply
  block has depth $(5 \times  2\kappa)\lg n$ (recall that the group program
  $\bbP$  has length  $n^{2\kappa}$).   Thus the  total  depth of  the
  UGP-Match  Circuit is  $(10\kappa+1)\lg  n  + \lg\lg  n  + 2$.   The
  Barrington Transform  (see Theorem \ref{thm:barrington-transform} of
  UGP-Match Circuit gives us a final group program under the FKN model
  of length $\tilde{O}(n^{20\kappa+2})$.

\end{proof}

\begin{corollary}
  UGP-Match  is an  efficient  and secure  protocol  for matching  any
  predicate in \NCone.
\end{corollary}

\subsection{Fixed Structure Group Programs}
We now state and prove the key lemma concerning fixed structure group programs. 

\begin{lemma}[Fixed structure programs yield 2-decomposable randomized
  encodings of universal functions]
\label{lemma:fixed-structure}
Fixed structure programs  are convertible to 2-decomposable randomized
encodings of universal functions with no loss of efficiency.
\end{lemma}
\begin{proof}
  The  conversion of  a fixed  structure program  to  a 2-decomposable
  randomized    encodings   of    universal   functions    is   fairly
  straightforward involving instantiation and coordinated blinding.

  Consider a class of predicates $\bbP$ convertible to fixed structure
  group      programs.       Let     $(g_1,g_2,\ldots,      g_{L+1}),$
  $(k_1,k_2,\ldots,k_L)$ be the fixed structure group program, i.e., a
  $(\alpha,1)$-preserving  group program with  a fixed  index sequence
  (the index sequence is the same independent of the specific function
  that the  group program is  computing though the  interstitial group
  elements $(g_1,  g_2,\ldots, g_{L+1})$ would depend  on the specific
  function). We now need to demonstrate two functions $\hat{f}_1(\bbm,
  r)$  and   $\hat{f}_2(\bbP,  r)$  satisfying   the  requirements  of
  Definition  \ref{def:2-decomp}, with  the  additional constraint  of
  universality, namely that $f(\bbm, \bbP) = \bbP(\bbm)$.

  Given  a   specific  instance   of  metadata  $\bbm$   the  function
  $\hat{f}_1$ first selects $\alpha$ or $1$ as appropriate for each of
  the index sequence pairs; note  that this is done independent of the
  specific  predicate $\bbP$.  Then $\hat{f}_1$  blinds  the resulting
  sequence  using  the randomness  $r$.  Similarly,  depending on  the
  specific  predicate   $\bbP$  $\hat{f}_2$  converts   to  the  fixed
  structure  group program and  gets a  specific instantiation  of the
  interstitial group elements $(g_1, g_2,\ldots, g_{L+1})$.  Then this
  sequence  is blinded  by $\hat{f}_2$  using $r$.   Observe  that the
  fixed structure  is crucial for  the construction of  $\hat{f}_1$ so
  that it is independent of $\bbP$. u
 
  It  remains to  prove  Correctness and  Security  as per  Definition
  \ref{def:2-decomp}.   First, the  existence of  REC  and Correctness
  follows because  of the appropriate cancelation of  the blinders and
  so  $\bbP(\bbm)$  can  be  recovered exactly  from  multiplying  the
  elements of the final  assembled group program.  Next, the existence
  of  SIM and  Security -  observe  that independent  of the  specific
  instance  the broker  receives shares  of  the same  form from  both
  publisher  and subscriber  and by  Lemma \ref{lemma:blindinglemma}
  (Blinding  Lemma) these  shares are  uniformly distributed  over all
  possible  sequences with  same  final value  and  therefore SIM  can
  generate the  outputs with the same distribution  as $\hat{f}_1$ and
  $\hat{f}_2$  by  generating all  but  one  element  of the  sequence
  completely and uniformly at random, then generating the last element
  so that the value of the sequence is the given final value, and then
  splitting  the group  program into  the two  parts  corresponding to
  $\hat{f}_1$ and $\hat{f}_2$.
\end{proof}

Observe that an implication of the above Lemma is that the fixed structure 
of the group program implicitly encodes a universal function. 

The   below   corollary   follows   from   the   two   lemmas,   Lemma
\ref{lemma:fixed-structure}  and  \ref{lemma:ccbps-2decompuniv}, since
2-decomposable   randomized  encodings   of  universal   functions  are
isomorphic to protocols for the C-CBPS model.

\begin{corollary}
\label{cor:fixed-structure}
  Fixed structure programs yield  secure C-CBPS protocols with no loss
  of efficiency.
\end{corollary}

\subsection{FSGP-Match}

With      Lemma     \ref{lemma:fixed-structure}      and     Corollary
\ref{cor:fixed-structure} in hand, all we need to do is to demonstrate
a fixed structure  group program of the right  complexity.  We give an
intuitive    description   of   the    main   idea.     Applying   the
$(\alpha,1)$-preserving Barrington  Transform to the  predicate $\bbP$
yields  an  $(\alpha,1)$-preserving  group  program but  not  a  fixed
structure  group  program  because   the  index  sequence  would  vary
depending  on  the predicate  $\bbP$.   So,  then  we substitute  each
instance of  a group  program pair  in this group  program by  a fixed
$(\alpha, 1)$-preserving group program block we call the Fixed Selector
Group  Program.   This creates  a  fixed  structure.   And by  padding
appropriately we can fix the length to be independent of the inputs as
well. It  remains to  describe the Fixed  Selector Group  Program. The
Fixed  Selector Group  Program is  obtained by  applying  an $(\alpha,
1)$-preserving   Barrington   Transform    to   the   Fixed   Selector
Circuit. Figure  \ref{fig:fixed-selector-circuit} is a  visual description
of the Fixed Selector Circuit. The  Fixed Selector Circuit is an OR of
ANDs.  It  has one  AND gate for  each bit  of the metadata  $\bbm$ as
input. The  other input  to the AND  gate is  a ``new'' bit  under the
control of  the subscriber. Since  the subscriber knows  the predicate
$\bbP$ it  knows exactly which bit  of the input ($\bbm$)  it needs in
that particular  block (based on the  index of the  group program pair
that the  block is replacing). Once the  subscriber instantiates these
bits approrpriately  and blinds then the Fixed  Selector Group Program
is just a  block with a fixed index sequence  since all the dependence
on  the  predicate $\bbP$  is  factored  into  the interstitial  group
elements.  Thus  the overall  index sequence is  fixed. Note  that for
blocks in the padding the subscriber  can set all new bits to $0$ thus
effectively reducing the padding part of the group program to a no-op,
something that evaluates to $1$.

\begin{figure}
\centering
\includegraphics[width=3.4in, height=3in]{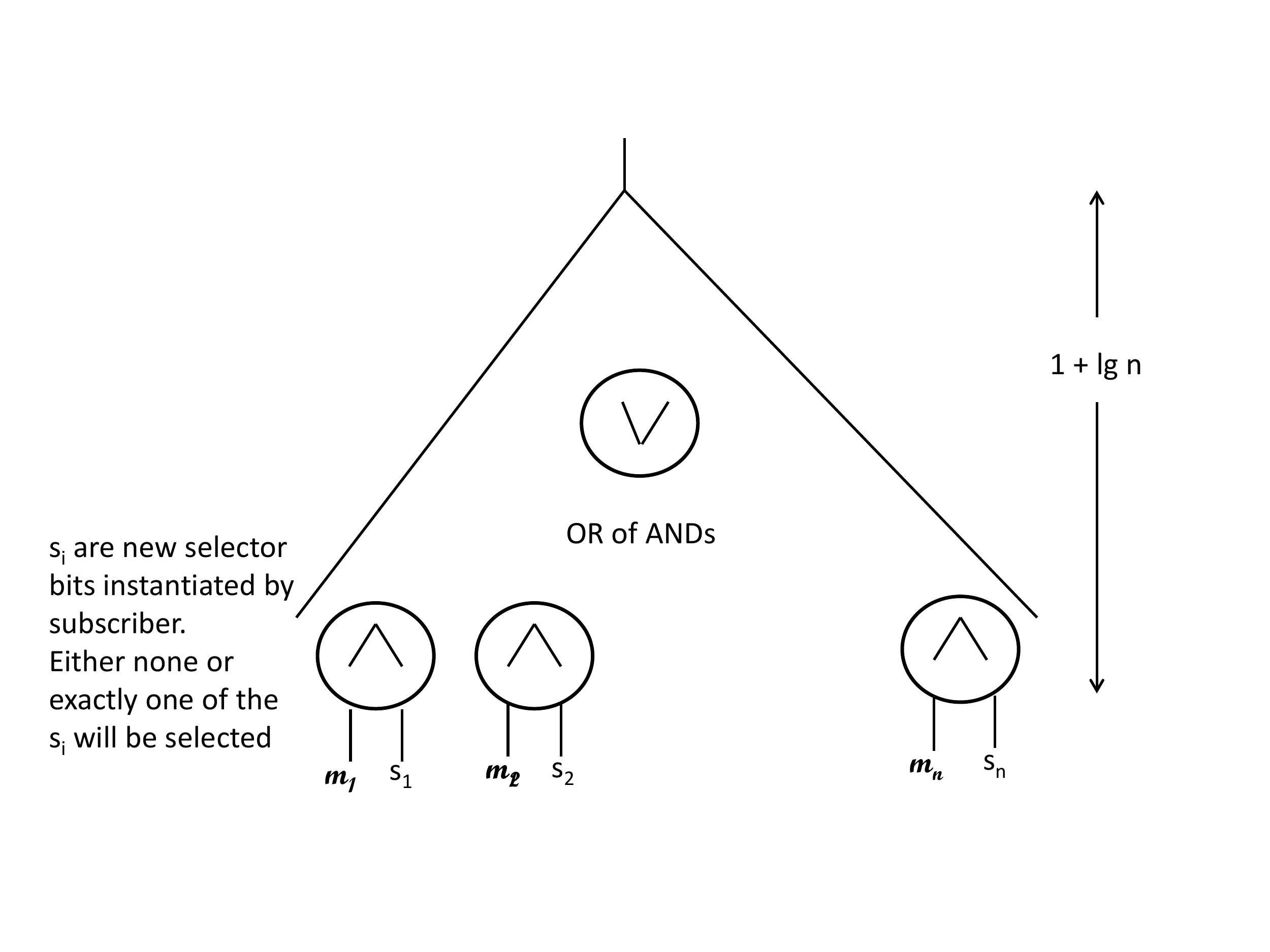}
\vspace{-.6in}
\caption{Fixed Selector Circuit}
\label{fig:fixed-selector-circuit}
\end{figure}

\begin{framed}
\underline{FSGP-Match}
\begin{enumerate}
\item  The  publisher and  subscriber  register  with  the broker  the
  precise  form of  their inputs.   In particular  the  publisher must
  specify the number $n$ of bits of $\bbm$ the meta-data (if there are
  fewer relevant  bits, the  remaining bits can  be padded  with dummy
  bits).  And the subscriber specifies  $D = 2\kappa\lg n$ the maximal
  depth of the predicate $\bbP$.

\item The  broker now  comutes the form  of the fixed  structure group
  program   by  applying   the  $(\alpha,   1)$-preserving  Barrington
  Transform to  any circuit of  depth $D$ to create  the corresponding
  $(\alpha, 1)$-preserving group program.  Again, only the form of the
  group  program is  relevant at  this stage,  not the  values  of the
  specific  interstitial  group  elements  or  indices  in  the  index
  sequence. It  then replaces each  group program pair with  the Fixed
  Selector  Group  Program  and  returns the  entire  composite  group
  program  back to  the publisher  and subscriber.  At this  point the
  values  of the  interstitial group  elements do  not matter  but the
  index  sequence, which  is fixed,  does  matter. The  program is  an
  $\alpha, 1)$-preserving  group program with indices  that index into
  the metadata bits.

\item The  publisher selects one of  $\alpha$ or $1$  depending on the
  value of the corresponding metadata bit for each group program pair,
  blinds them  and sends back  to the broker. The  subscriber, applies
  the  $(\alpha, 1)$-preserving Barrington  Transform to  the specific
  instance  $\bbP$ to  get the  concrete values  for  the interstitial
  group elements,  appropriately sets the  ``new'' bits for  the Fixed
  Selector Group  Program blocks, and  blinds the entire  sequence and
  sends to the broker.

\item The broker  puts all the blinded elements  together in the right
  sequence and multiplies them.  If  he gets $1$ he forwards along the
  (encrypted)  data from  the  publisher to  the  subscriber, else  he
  withholds it.
\end{enumerate}
\end{framed}

\begin{theorem}
  \label{thm:FSGP}
  Given metadata  of size at most  $n$ and predicate of  depth at most
  $\kappa\lg  n$, FSGP-Match  is  an information-theoretically  secure
  protocol for the C-CBPS model with complexity $4n^{2\kappa+2}$.
\end{theorem}
\begin{proof}
  FSGP-Match  is  obtained  by  applying  the  $(\alpha,1)$-preserving
  Barrington   Transform  to   the  predicate   $\bbP$  to   yield  an
  $(\alpha,1)$-preserving group program where each instance of a group
  program pair is substituted by the Fixed Selector Group Program.  It
  is clear  that FSGP-Match produces  a fixed structure  group program
  and hence it follows  from Lemma \ref{lemma:fixed-structure} that it
  is correct and secure.

  All that  remains to  do is to  bound the complexity  of FSGP-Match.
  Recall  that group  program pair  in  the group  program (of  length
  $n^{2\kappa}$   is   substituted  by   the   Fixed  Selector   Group
  Program.   The   Fixed    Selector   Group   circuit   (see   Figure
  \ref{fig:fixed-selector-circuit}) has  depth $1 +  \lg n$ and  hence the
  resulting Fixed Selector Group Program has length $4n^2$ which means
  that the final group program has a length $4n^2 \times n^{2\kappa} =
  4n^{2\kappa + 2}$.
\end{proof}

\subsection{OFSGP-Match}
OFSGP-Match is basically FSGP-Match with the additional innovation being that we use an 
Optimized Fixed Selector Group Program. See Figure \ref{fig:OFSGP}.

\begin{figure}
\centering
\includegraphics[width=3.4in, height=3in]{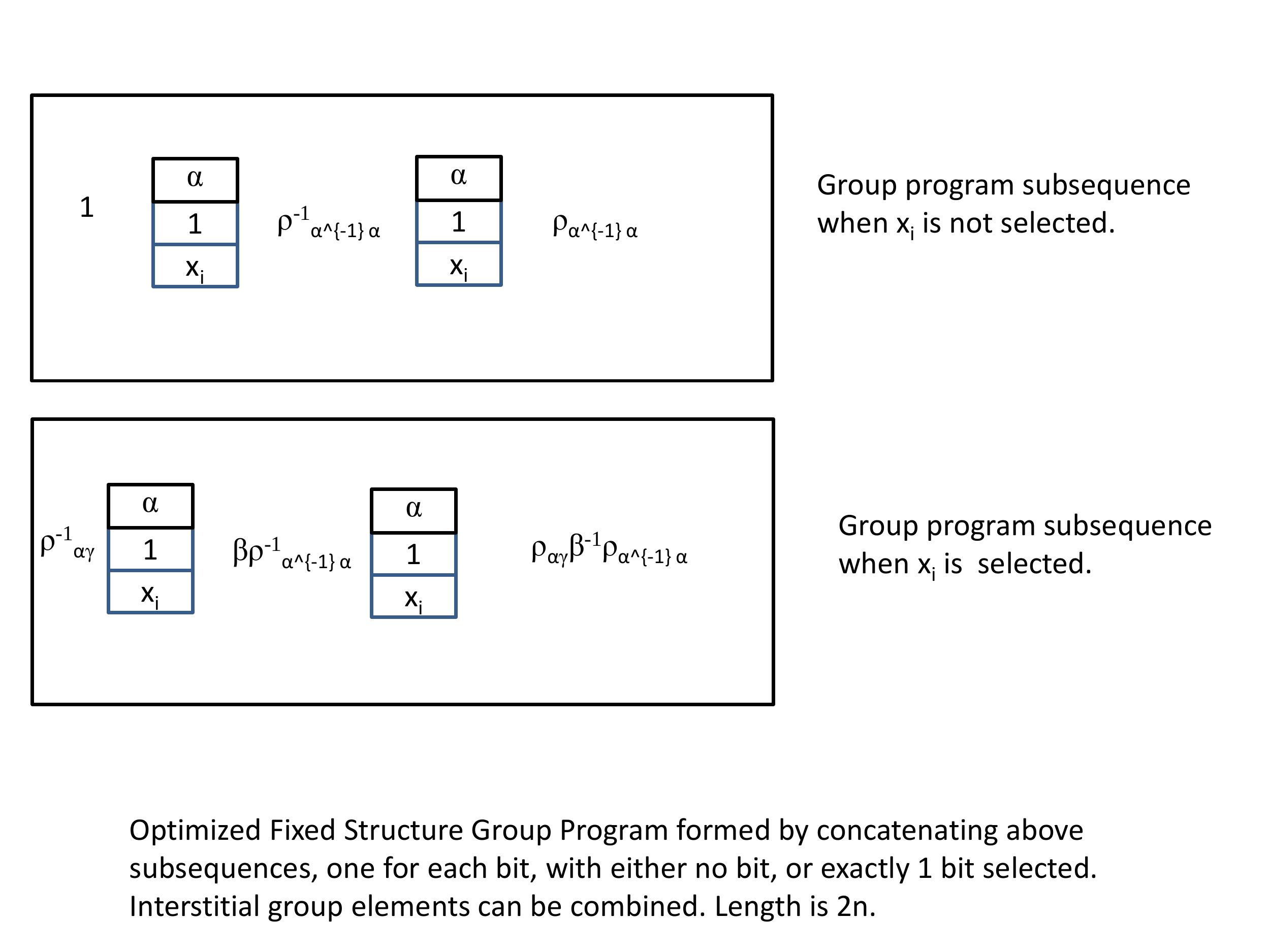}
\caption{(Hand) Optimized Fixed Selector Group Program}
\label{fig:OFSGP}
\end{figure}

\begin{theorem}
  Given metadata  of size at most  $n$ and predicate of  depth at most
  $\kappa\lg  n$,  OFSGP-Match  is an  information-theoretically  secure
  protocol     for     the     C-CBPS    model     with     complexity
  $2n^{2\kappa+1}$.
\end{theorem}
\begin{proof}
The proof is similar to the proof of Theorem \ref{thm:FSGP} so we do not repeat it. 

The Optimized Fixed Selector Group Program (see Fig. \ref{fig:OFSGP}) has length $2n$ which means
  that the final group program has a length $2n \times n^{2\kappa} =
  2n^{2\kappa + 1}$.
\end{proof}

\section{Implementation and Performance}
\label{sec:implementation}
We implemented all three of our protocols in Java. We will spare the reader the details of the code. As has already been explained the protocols are conceptually straightforward. The only protocol that was somewhat performant was OFSGP-Match. In particular, we considered the Hamming problem as that is representative of a typical function in the context of publish/subscribe. The publisher and the subscriber each have a private bit-vector of $n$ bits. They wish to compute a secure Hamming match, i.e., given a common threshold level (between $0$ and $\log_2 n$) they wish to exchange the actual content iff the Hamming distance between their respective vectors exceeds the threshold. The bottleneck from a computational standpoint is the cost of computing a match at the broker, which is basically the length of the resulting group program. Since, in our setting the subscribers submit predicates, we have our subscriber submit the Hamming distance circuit with its own bit-vector and the threshold hard-coded into it.

Our match algorithm implemented as a circuit, essentially, involved computing the difference vector of the two vectors, then sorting the bits using a Bose-Nelson sorting network \cite{BN62}, and outputting the threshold bit as the match. We used a Bose-Nelson circuit because they have the shortest depth (i.e. best constants) for small depths (even though AKS sorting networks \cite{AKS83} are asymptotically optimal - depth $O(\log n)$ - instead of $O(\sqrt(n))$ for Bose-Nelson).
We implemented the above match algorithm on a 3GHz quad-core i7 laptop. Our goal was to find the largest $n$ for which the match could be computed under 1s.
The results are presented in Table \ref{tab:hamming} in Appendix \ref{appendix:hamming}. Here $n$ is the length of the bit-vectors. $d$ is the depth of the resulting ciruit and the final column shows the length of the resulting OFSGP-Match sequence. The main takeaway is that even our best scheme becomes unusable for a relative simple function such as the Hamming distance once the bit vectors get to over 16 bits in length.

\junk{
*** KRASH - Talk about implementation details, architecture etc etc ***
This section describes the prototype library implementation of the SPAR-Match Algorithm, which corresponds to Task 5.1 of the WBS.  The implementation includes the following components.  Please refer to the companion code listing provided separately for low-level details of the implementation.
Generic Utilities:  Several utility functions including support for higher order functions from functional programming, list manipulation and shorthand are included in the prelude.

S5 Permutation Group Functions:  We provide functions to generate and represent members of the S5 permutation group.  7-bit, 15-bit, and a vector representation are supported.  Multiplication and inverse operations in S5 derived from first principles as well as a faster table look-up implementation are supported.

Barrington.s Transform:  We define S5 members (including the alpha and beta cycles, and other constants) used in Barrington.s transformation.  We implement the transformation for AND gates and NOT gates and a parser for an Acyclic Boolean Circuit comprised of AND gates and NOT gates in S-expression form.

SPAR-Match Algorithm: We implement the steps of the SPAR-Match algorithm performed by the publisher, subscriber and broker.  The publisher takes the raw input bits and prepares the publisher group elements and then blinds them using the Feige-Killian-Naor protocol to send to the broker for a secure match. The implementation converts the circuit to a group program given a subscriber predicate in circuit form along with any local constants. The group program is padded to the maximum length, publisher bit references are replaced with selector blocks, and the resulting group program is canonicalized to a form where alternating elements are from publisher and subscriber by multiplying the subscriber elements as applicable. The resulting Universal Group Program is blinded using the Feige-Killian-Naor protocol to send to the broker.  The broker interleaves the sequences from the publisher and subscriber and then multiplies them to evaluate a match (true if the result equals S5 member alpha).
A random number generator based on SHA1PRNG from the Java security library is used by our implementation.   A unit test to exercise the publisher, subscriber and broker functionality is provided. The solution is general enough to support any subscriber predicate within the circuit complexity class NC1.

Practical Pub/Sub Illustration: A simple end-to-end publish/subscribe functionality using SPAR-Match is illustrated in the implementation (see the companion code listing for details).  The illustration includes syntax definitions, example user-readable metadata and user-readable subscription predicates in S-expression format. Examples include intelligence report records, and stock ticker records.  We include functionality to parse and encode the metadata and subscriber predicates to bits and circuits that SPAR-Match can use.  A unit test that exercises the end-to-end flow is included.  The illustration is not intended to comprehensively cover everything that can be done with the SPAR-Match. From the illustration it is easy to see how more sophisticated parsers and a wider range of predicates within NC1 can be readily supported with additional implementation effort. 

Integration and Test Plan: In preparation for integration and internal test and evaluation, we have identified insertion points for SPAR-Match within the Siena system. HierarchicalDispatcher is the class within Siena that is responsible for managing subscriptions (using the subscribe() method) and forwarding publications (using the publish() method). We will replace the HierarchicalDispatcher match logic with an implementation of our algorithm. Figure 10 shows the software organization, covering Carzaniga and Wolf.s HierarchicalDispatcher in the siena package, Raiciu and Rosenblum.s OperatorPublisher and OperatorSubscriber classes in the securematching package, our new HierarchicalDispatcher in the spar package and our new Client, Producer and Consumer classes in the testfixture package.
}

\junk{
\section{Performance}

\section{Limitations}
\label{sec:limitations}
In theory, we can use such single message protocols
in the context of publish/subscribe, but with tremendous overhead:
For every published notification, the publisher and all subscribers
would generate a new key, the subscribers would then
register their subscriptions, and finally the publisher would send
the notification. Secondly, even the cheapest instances of these
protocols has high costs for single invocations.

\junk{
\section{Leakage}

Here is a list of exactly what is leaked in our system/protocol:

\begin{itemize}
\item Identities of all the publishers must be known to each subscriber and vice versa. Of course
the broker must know the identities of all.

\item Number of pieces of content produced by each publisher in any given time period - this will be known
to everybody, in particular each subscriber will have maintain their index in the
one-time pad they share with each publisher and produce a blinded $BG-C$ to provide
to the broker. Further this can also be monitored by any external spectator. Note that
this is not specific to our protocol but to any protocol, unless the production of 
content is obfuscated through some such technique as batching or generation of chaff content etc.

\item Maximum size of the data produced by any publisher. Since we standardize on
message lengths for the protocol across all publishers and subscribers, by looking
at the bits exchanged any external spectator will be able to tell the maximum data size.

\item Maximum size of the predicate of any subscriber. Again, since we standardize on
message lengths for the protocol across all publishers and subscribers, by looking
at the bits exchanged any external spectator will be able to tell the maximum circuit size.
\end{itemize}

No information other than the above is leaked.
}
}

\section{Conclusion}
\label{sec:conclusion}
We started out on a research plan to achieve secure publish/subscribe at line speeds. Though we did not reach our goal this paper reports on the substantial progress we achieved - namely information-theoretically secure publish/subscribe for predicated in \NCone. On the theoretical front we achieved a level of expressivity that had not been previously attained and on the practical front we made a substantial advance towards a practical scheme. Our protocols have the benefit of being conceptually clean and simple, so simple that we can capture their complexity without even employing the big-Oh notation, see Table \ref{tab:complexity} . Going forward, the main open question is to come up with a much faster protocol, potentially one that will scale at wireline speeds (Gbps). 

\newpage

\bibliographystyle{abbrv}
\bibliography{SPAR-Match}

\appendix
\section{Barrington Transform}
\label{appendix:barrington}
\subsection{Barrington Transform}
\label{appendix:barrington-transform}
We present  the proof of  Theorem \ref{thm:barrington-transform} based
on the  treatment in \cite{Viola09}.  (Note that the statement  of the
theorem is true for any cycle $\alpha$ and not just the specific cycle
$(2 3 4 5 1)$.

\begin{proof}

We present the proof as a series of lemmas.

\begin{lemma}[The cycle does not matter]
\label{lemma:cycle-doesnt-matter}
  Let $\alpha  , \beta  \in S_{5}$ be  two cycles,  let $f:\{0,1\}^{n}
  \rightarrow  \{0,1\}$.  Then f  is  $\alpha$-computable with  length
  $\ell$ $\Leftrightarrow$ f is $\beta$-computable with length $\ell$.
\end{lemma}
\begin{proof}
  First observe that $\exists \rho_{\alpha\beta} \in S_{5} $ such that
  $\alpha = \rho_{\alpha\beta}^{-1} \beta \rho_{\alpha\beta}$.

  To   see   this   let   \[\alpha  =   (\alpha_{1},\alpha_{2},   ...,
  \alpha_{5})_{\mbox{cycle}}\]
\[\beta  = (\beta_{1},\beta_{2}, ..., \beta_{5})_{\mbox{cycle}}\]
\[\rho_{\alpha\beta}  := (\alpha_{1}\rightarrow \beta_{1},\alpha_{2}\rightarrow \beta_{2},...,\alpha_{5}\rightarrow \beta_{5}).\]
(With $\alpha = (2 3 4 5  1)$ and $\beta = (3 5 4 2 1)$ we
get $\rho_{\alpha\beta} = (1 3 4 2 5), \rho_{\alpha\beta}^{-1} = (1 4 2 3 5)$.

Suppose                                                            that
$(g_{1}^{0},...,g_{\ell}^{0})(g_{1}^{1},...,g_{l}^{1})(k_{1},...k_{\ell})$
$\beta$-computes $f$; we  claim that $(\rho_{\alpha\beta}^{-1} g_{1}^{0},...,g_{\ell}^{0}
\rho_{\alpha\beta})(\rho_{\alpha\beta}^{-1} g_{1}^{1},...,g_{\ell}^{1}  \rho_{\alpha\beta})$ (with the same
indices $k_i$) $\alpha$-computes $f$. To see this, note that
\[\prod^{\ell} _{i=1} g_{i}^{x_{k_i}}= 1_{G}  \Rightarrow \rho_{\alpha\beta}^{-1} \prod ^{l} _{i=1} g_{i}^{x_{k_i}} \rho = \rho^{-1} \cdot \rho_{\alpha\beta} = 1,\]
\[\prod ^{\ell} _{i=1} g_{i}^{x_{k_i}}= \beta \Rightarrow \rho_{\alpha\beta}^{-1} \prod ^{l} _{i=1} g_{i}^{x_{ki}} \rho = \rho^{-1} \beta \rho_{\alpha\beta} = \alpha.\]
\end{proof}

\begin{lemma} [$f \Rightarrow 1-f$]
  If $f:\{0,1\}^{n}  \rightarrow \{0,1\}$ is  $\alpha$-computable by a
  group program of length $\ell$, so is $1-f$.
\end{lemma}
\begin{proof}
  First apply  the previous  lemma to $\alpha^{-1}$-compute  $f$. Then
  multiply last  group elements  $g^0_{\ell}$ and $g^1_{\ell}$  in the
  group program by $\alpha$. (Note that $\alpha^{-1} = (5 1 2 3 4)$.)
\end{proof}

\begin{lemma}[$f,g \Rightarrow f \wedge g$]
\label{lemma:4times}
  If $f$ is $\alpha$-computable with  length $\ell$ and $g$ is $\beta$
  computable    with   length   $\ell$    then   ($f\wedge    g$)   is
  ($\alpha\beta\alpha^{-1}\beta^{-1}$)-computable with length $4\ell$.
\end{lemma}
\begin{proof}
  Concatenate  $4$ programs: ($\alpha$-computes  $f$, $\beta$-computes
  $g$,   $\alpha^{-1}$-computes   $f$,   $\beta^{-1}$-computes   $g$).
  (f(x)=1)$\wedge$    (g(x)=1)$\Rightarrow$    concatenated    program
  evaluates  to  ($\alpha\beta\alpha^{-1}\beta^{-1}$);  but if  either
  $f(x) = 0$ or $g(x) =  0$ then the concatenated program evaluates to
  $0$. For example,  if $f(x)=0$ and $g(x) =  1$ then the concatenated
  program gives $1\cdot\beta\cdot1\cdot\beta^{-1} = 1$.
\end{proof}

It only  remains to  see that  we can apply  the previous  lemma while
still computing with respect to a cycle.

\begin{lemma}
\label{lemma:cycle-exists}
  $\exists       \alpha,      \beta$       cycles       such      that
  $\alpha\beta\alpha^{-1}\beta^{-1}$ is a cycle.
\end{lemma}
\begin{proof}
  Let $\alpha  :=(2 3 4 5  1)$, $\beta := (3  5 4 2 1)$,  we can check
  $\alpha\beta\alpha^{-1}\beta^{-1} = \gamma = (3 5 2 1 4 )$ is a cycle.
\end{proof}

The proof of the theorem follows by induction on $d$ using previous lemmas.
\end{proof}

\subsection{$(\alpha, 1)$-preserving Barrington Transform}
\label{appendix:alpha-one-barrington}
The proof of  Theorem \ref{thm:alpha-one-barrington-transform} is basically
the same as the proof of Theorem \ref{thm:barrington-transform} with 
minor modifications to account for the fact that we are now transforming the
circuit into an $(\alpha, 1)$-preserving group program.

\begin{proof}

As before we present the proof as a series of lemmas.

\begin{lemma}[The cycle does not matter]
  Let $\alpha  , \beta  \in S_{5}$ be  two cycles,  let $f:\{0,1\}^{n}
  \rightarrow  \{0,1\}$.  Then f  is  $\alpha$-computable with  length
  $\ell$ $\Leftrightarrow$ f is $\beta$-computable with length $\ell$.
\end{lemma}
\begin{proof}
The proof is similar to that of Lemma \ref{lemma:cycle-doesnt-matter}.
Let $\alpha  , \beta  \rho_{\alpha\beta} \in S_5$ be as in that lemma.  
Suppose  that $(g_1,  g_2,\ldots,  g_{L+1})$ $(k_1,  k_2, \ldots,  k_L)$
$\beta$-computes $f$ then it  follows, in straightforward fashion that
$(\rho_{\alpha\beta}^{-1}g_1,  g_2,\ldots,  g_{L+1}\rho_{\alpha\beta})$
$(k_1, k_2, \ldots, k_L)$ $\alpha$-computes $f$.
\end{proof}

\begin{lemma} [$f \Rightarrow 1-f$]
  If $f:\{0,1\}^{n}  \rightarrow \{0,1\}$ is  $\alpha$-computable by a
  group program of length $\ell$, so is $1-f$.
\end{lemma}
\begin{proof}
  First apply  the previous  lemma to $\alpha^{-1}$-compute  $f$. Then
  multiply the last  group elements  $g_L$ 
  group program by $\alpha$. 
\end{proof}

\begin{lemma}[$f,g \Rightarrow f \wedge g$]
  If $f$ is $\alpha$-computable with  length $\ell$ and $g$ is $\beta$
  computable    with   length   $\ell$    then   ($f\wedge    g$)   is
  ($\alpha\beta\alpha^{-1}\beta^{-1}$)-computable with length $4\ell$.
\end{lemma}
\begin{proof}
The proof of this Lemma is identical to the proof of Lemma~\ref{lemma:4times}.
\end{proof}

The proof of the theorem follows by induction on $d$ using previous lemmas.
and Lemma~\ref{lemma:cycle-exists}.
\end{proof}

\junk{
\subsection{Explicit values}

Here I explicitly list a set of values that work: 

$1_{S_5} = (1 2 3 4 5)$

$\alpha = (2 3 4 5 1)$ and $\alpha^{-1} = (5 1 2 3 4)$

$\rho_{\alpha\alpha^{-1}} = (1 5 4 3 2) = \rho_{\alpha\alpha^{-1}}^{-1}$ and $\alpha = \rho_{alpha\alpha^{-1}}^{-1} \alpha^{-1} \rho_{\alpha\alpha^{-1}}$

$\beta = (3 5 4 2 1)$ and $\beta^{-1} = (5 4 1 3 2)$

$\rho_{\alpha\beta} = (1 3 4 2 5)$ and $\beta = \rho_{\alpha\beta}\cdot \alpha\cdot \rho^{-1}_{\alpha\beta}$

$\alpha\cdot\beta\cdot\alpha^{-1}\cdot\beta^{-1} = (3 5 2 1 4) = \gamma$

$\rho_{\alpha\gamma} = (1 3 2 5 4)$ and $\gamma = \rho_{\alpha\gamma}\cdot \alpha\cdot \rho^{-1}_{\alpha\gamma}$
}

\clearpage
\newpage
\section{Hamming distance - Performance}
\label{appendix:hamming}
\begin{table}[h]
\centering
\begin{tabular}{||c|c|c||}
\hline \hline
$n$ - bit-vector length & $d$ - circuit depth & OFSGP-Match  - sequence length\\ \hline \hline
$2$ & $5$ & $4096$ \\ \hline
$3$ & $8$ & $393216$ \\ \hline
$4$ & $8$ & $524288$ \\ \hline
$5$ & $12$ & $1.68* 10^8$ \\ \hline
$6$ & $12$ & $2.01*10^8$ \\ \hline
$7$ & $13$ & $9.4*10^8$ \\ \hline
$8$ & $13$ & $1.07*10^9$ \\ \hline
$9$ & $16$ & $7.73*10^{10}$\\ \hline
$10$ & $16$ & $8.59*10^{10}$ \\ \hline
$11$ & $16$ & $9.45*10^{10}$ \\ \hline
$12$ & $16$ & $1.03*10^{11}$ \\ \hline
$13$ & $16$ & $1.12*10^{11}$ \\ \hline
$14$ & $16$ & $1.2*10^{11}$ \\ \hline
$15$ & $16$ & $1.29*10^{11}$ \\ \hline
$16$ & $16$ & $1.37*10^{11}$ \\ 
\hline \hline
\end{tabular}
\caption{Length of group program for computing Hamming distance.}
\label{tab:hamming}
\end{table}

\end{document}